\documentclass[twocolumn,twoside]{IEEEtran}
\usepackage{amsmath,amsfonts}
\usepackage{algorithm}
\usepackage{algpseudocode}
\usepackage{amsmath} 
\usepackage{array}
\usepackage{subfigure}
\usepackage{textcomp}
\usepackage{stfloats}
\usepackage{url}
\usepackage{verbatim}
\usepackage{graphicx}
\usepackage{cite}
\usepackage{amsthm,amssymb}
\usepackage{color}
\usepackage{pifont}
\usepackage{hyperref}
\usepackage{footnote}
\usepackage{enumitem}
\usepackage{xcolor}
\usepackage{diagbox}
\hypersetup{
colorlinks=true,
linkcolor=black,
citecolor=black,
}

\graphicspath{{Figures/}}

\newtheorem{theorem}{Theorem}[section]
\newtheorem{proposition}[theorem]{Proposition}
\newtheorem{lemma}[theorem]{Lemma}

\theoremstyle{definition}
\newtheorem{remark}[theorem]{Remark}

\newcommand{\longthmtitle}[1]{\mbox{}{\textit{(#1):}}}

\newcommand{\real}{\ensuremath{\mathbb{R}}}
\newcommand{\complex}{\ensuremath{\mathbb{C}}}

\newcommand{\Cc}{\mathcal{C}}
\newcommand{\Dc}{\mathcal{D}}
\newcommand{\Ec}{\mathcal{E}}

\newcommand{\Gc}{\mathcal{G}}

\newcommand{\Lc}{\mathcal{L}}

\newcommand{\Nc}{\mathcal{N}}

\newcommand{\Qc}{\mathcal{Q}}

\newcommand{\Xc}{\mathcal{X}}

\newcommand\abf{\mathbf{a}}
\newcommand\bbf{\mathbf{b}}

\newcommand\gbf{\mathbf{g}}
\newcommand\hbf{\mathbf{h}}
\newcommand\ibf{\mathbf{i}}

\newcommand\pbf{\mathbf{p}}
\newcommand\qbf{\mathbf{q}}

\newcommand\sbf{\mathbf{s}}

\newcommand\ubf{\mathbf{u}}
\newcommand\vbf{\mathbf{v}}
\newcommand\wbf{\mathbf{w}}

\newcommand\ybf{\mathbf{y}}

\newcommand\Abf{\mathbf{A}}
\newcommand\Bbf{\mathbf{B}}

\newcommand\Dbf{\mathbf{D}}
\newcommand\Ebf{\mathbf{E}}

\newcommand\Ibf{\mathbf{I}}

\newcommand\Mbf{\mathbf{M}}

\newcommand\Rbf{\mathbf{R}}

\newcommand\Xbf{\mathbf{X}}
\newcommand\Ybf{\mathbf{Y}}
\newcommand\Zbf{\mathbf{Z}}

\newcommand\Gammab{\boldsymbol{\Gamma}}
\newcommand\phib{\boldsymbol{\phi}}
\newcommand\rhob{\boldsymbol{\rho}}
\newcommand\varrhob{\boldsymbol{\varrho}}
\newcommand\xib{\boldsymbol{\xi}}

\newcommand\sign{\rm{sign}}

\newcommand{\ones}{\mathbf{1}}
\newcommand{\zeros}{\mathbf{0}}

\newcommand{\diag}{{\rm diag}}

\newcommand\range{\text{range}}

\newcommand{\oprocendsymbol}{\hbox{$\bullet$}}
\newcommand{\oprocend}{\relax\ifmmode\else\unskip\hfill\fi\oprocendsymbol}

\renewcommand{\baselinestretch}{.98}
\allowdisplaybreaks

\renewcommand{\algorithmicrequire}{\textbf{Input:}}

\makeatletter
\def\footnoterule{\kern-3\p@
  \hrule \@width 2in \kern 2.6\p@} 
\makeatother

\definecolor{navy}{HTML}{000080}

\begin{document}

\title{Learning Provably Stable Local Volt/Var Controllers for Efficient Network Operation
}

\author{Zhenyi Yuan,~\IEEEmembership{Student Member,~IEEE}, Guido Cavraro,~\IEEEmembership{Member,~IEEE},\\ Manish K. Singh,~\IEEEmembership{Member,~IEEE}, and Jorge Cort\'es,~\IEEEmembership{Fellow,~IEEE}

\thanks{This work was authored by the National Renewable Energy Laboratory, operated by Alliance for Sustainable Energy, LLC, for the U.S. Department of Energy (DOE) under Contract No. DE-AC36-08GO28308. Funding provided by the NREL Laboratory Directed Research and Development Program. The views expressed in the article do not necessarily represent the views of the DOE or the U.S. Government. The U.S. Government retains and the publisher, by accepting the article for publication, acknowledges that the U.S. Government retains a nonexclusive, paid-up, irrevocable, worldwide license to publish or reproduce the published form of this work, or allow others to do so, for U.S. Government purposes. This work was also partially supported by NSF Award ECCS-1947050.}
\thanks{Z. Yuan and J. Cort\'es are with the Department of Mechanical and Aerospace Engineering, UC San Diego. {\tt \{z7yuan,cortes\}@ucsd.edu}. G. Cavraro is with the Power Systems Engineering Center, National Renewable Energy Laboratory. {\tt guido.cavraro@nrel.gov}. M. K. Singh is with the Department of Electrical and Computer Engineering, University of Minnesota. {\tt msingh@umn.edu.}}
}

\markboth{IEEE Transactions on Power Systems}%
{Yuan \MakeLowercase{\textit{et al.}}: Learning Provably Stable Local Volt/Var Controllers for Efficient Network Operation}


\maketitle

\begin{abstract}
This paper develops a data-driven framework to synthesize local Volt/Var control strategies for distributed energy resources (DERs) in power distribution grids (DGs). Aiming to improve DG operational efficiency, as quantified by a generic optimal reactive power flow (ORPF) problem, we propose a two-stage approach. The \emph{first} stage involves learning the manifold of optimal operating points determined by an ORPF instance.
To synthesize local Volt/Var controllers, the learning task is partitioned into learning local surrogates (one per DER) of the optimal manifold with voltage input and reactive power output. Since these surrogates characterize efficient DG operating points, in the \emph{second} stage, we develop local control schemes that steer the DG to these operating points. We identify the conditions on the surrogates and control parameters to ensure that the locally acting controllers collectively converge, in a global asymptotic sense, to a DG operating point agreeing with the local surrogates. We use neural networks to model the surrogates and enforce the identified conditions in the training phase. AC power flow simulations on the IEEE 37-bus network empirically bolster the theoretical stability guarantees obtained under linearized power flow assumptions. The tests further highlight the optimality improvement compared to prevalent benchmark methods.
\end{abstract}


\section{Introduction}\label{sec:intro}

The deployment of a massive number of distributed energy resources (DERs) in power distribution grids (DGs) is dramatically changing the electric power grid. Primarily driven by sustainability and economic incentives, DERs present additional opportunities, including reductions in the power generation costs and of greenhouse gas emissions. Nevertheless, DERs' uncoordinated power injections or sudden generation changes could pose challenges to system operations and stability, e.g., induce undesirable voltage deviations in DGs.
To facilitate their integration in power grids, DERs are being provided with sensing and computational capabilities, hence becoming \emph{smart agents}. DERs can exploit the flexibility of their power electronic interface to perform, among other ancillary services, reactive power control. They can also take advantage of the widespread availability of data from DGs and the increased capabilities for storing and processing the data to learn effective control policies.
This paper aims to leverage learning in the synthesis of local Volt/Var controllers for voltage regulation incorporating optimality considerations and rigorous performance guarantees.

\subsubsection*{Literature Review}

The main goal of Volt/Var control strategies is to keep voltages within safe preassigned limits by commanding DERs' reactive power injections.
Classically, DERs' reactive power outputs are computed, in an \emph{open-loop} fashion, by the system operator solving \emph{optimal power flow} (OPF) problems. 
Efficient and advanced solvers for OPF problems are available, see, e.g.~\cite{SHL:14,BC-XAS:18,LG-NL-UT-SL:13}.
However, high penetration of renewable generation and increased variability of DGs require solving numerous instances of OPF problems within a limited time frame.
Aiming to tackle this challenge, several learning-based approaches have been proposed to predict OPF solutions, 
see, e.g.,~\cite{XP-TZ-MC-SZ:21,DO-FG-AR:20,MKS-SG-VK-GC-AB:20,MKS-VK-GBG:21,FF-TWKM-PVH:20}, to mention a few.
Once trained, the inference time for these approaches when presented with a new input is minimal.
Nevertheless, the solution of OPF problems requires information from \emph{all} the buses.
Specifically, power demands from loads and generation limits from generators must be precisely known. Such requirements are prohibitive for practical DGs because, in general, not all the buses are monitored in real time, individual loads are unlikely to announce their demand profiles in advance, and the availability of small size generators is hard to predict.

This has motivated the development of \emph{closed-loop} strategies, which compensate for the lack of information with measurements retrieved from the field.
Given the massive number of controllable devices envisioned to be hosted in future DGs, \emph{decentralized} approaches are often advocated for practical applications.
There are two notable classes of decentralized algorithms. The first consists of \emph{distributed} algorithms in which agents cooperate and share information with peers. Distributed algorithms can achieve optimal performance in the sense that they can be designed to exactly solve a given OPF instance, see, e.g.,~\cite{ED-HZ-GBG:13}.
Optimization-based feedback controllers that steer the network toward
solutions to OPF problems based on the cyclical alternation of sensing, communication, and actuation have recently become popular~\cite{GC-RC:17,ED-AS:18,GQ-NL:20}.
Nevertheless, distributed strategies are suitable for systems endowed with a reliable real-time communication network meeting precise and strict requirements, which are rarely satisfied in practice for DGs.
For instance, in many works, each generator is required to share information with all its neighbors in the power network before every power output update~\cite{GC-RC:17}.
The second class of decentralized algorithms consists of \emph{local} approaches, in which each agent makes decisions based only on information available locally.
In local schemes, reactive power compensations are adjusted based merely on measurements taken locally.
Even though pertinent standards allow DERs to provide reactive power compensations following static Volt/Var control rules, see IEEE standard 1547~\cite{IEEE1547}, the literature has provided a variety of options for local voltage regulation~\cite{GC-RC:17,KT-PS-SB-MC:11,HZ-HJL:15,XZ-MF-ZL-LC-SHL:21}.
However, local schemes have intrinsic performance limitations, e.g., they might fail to regulate voltages even if the overall generation resources are enough~\cite{SB-RC-GC-SZ:19}.

Recent advances in data-driven and learning-based control seek to leverage data from the plant to learn optimal controllers. Though impressive results have been demonstrated, it has not yet been widely used in engineering practice due to the lack of closed-loop stability guarantees.
Various results have been developed to tackle this problem, see e.g., a comprehensive review~\cite{CD-SG-CF:22}. Although most of these works achieve the stability of neural network control by using penalty functions to integrate the stability requirement as \emph{soft} constraints, recent work in power systems, primarily on frequency control, seeks to explicitly engineer the neural network structure to integrate the stability requirement as \emph{hard} constraints, see, e.g.,~\cite{WC-YJ-BZ:22,ZY-CZ-JC:22-scl}.
In the context of Volt/Var control, the goal is to leverage learning techniques to enhance the performance of local control schemes and reduce the gap with distributed and/or optimal controllers, while retaining closed-loop system stability. 
Related works include~\cite{WC-JL-BZ:22} and~\cite{YS-GQ-SL-AA-AW:22}, where 
reinforcement learning is used to learn stability-guaranteed local Volt/Var control schemes.
However, the former enforces stringent derivative constraints on the policies to be searched, whereas the latter only guarantees that the voltages converge to a region, instead of an equilibrium point, and both of them require the control policy to be continuously differentiable. Furthermore, neither of them takes accounts for the reactive power capacity limitations, which are critical when dealing with small-size generators.
Recent works~\cite{CZ-YX-YW-ZYD-RZ:21,SG-SC-VK:22} optimize the local Volt/Var control schemes though  interval optimization with a day-ahead schedule, while preserving closed-loop system stability. However, the considered local Volt/Var curves are limited to a piecewise linear form, and might suffer non-negligible optimality gaps.
Other works provide interesting insights on learning Volt/Var rules, but do not assess the stability of the overall system and hence are not straightforwardly suitable for practical applications:~\cite{SK-PA-GH:19,PST-PHG:22} leverage segmented linear regression techniques to learn local surrogates that predict OPF solutions,~\cite{XS-JQ-JZ:21} proposes to learn the local controller by taking as an input both voltages and active power setpoints, and~\cite{HJ-CW-PL-JZ-GS-FD-JW:18} proposes a framework for tuning the parameters of standard piecewise linear local voltage regulators.

\subsubsection*{Statement of Contributions}
We propose a framework for designing local Volt/Var scheme whose goal is to not only regulate voltages but also act as local surrogates of \emph{optimal reactive power flow} (ORPF) problem solvers. ORPF problems are particular instances of OPF problem in which the goal is to optimize the generator's reactive power injections.
We base our work on the distinction between the \emph{control function} and the \emph{equilibrium function}. The first represents the reactive power update rule, and the latter describes the possible system equilibrium points. In many works, the equilibrium functions coincide with the control functions, but this is not the case for our framework.
We advocate for a two-stage strategy. In the first stage, for each controllable node, the equilibrium function
providing the ORPF solution surrogates is learned from historical data. Precisely, such a function receives as input the local voltage and gives as an output an approximation of the optimal reactive power setpoint.
In the second stage, we devise a control function whose equilibrium points are exactly the ORPF approximated solutions provided by the equilibrium functions. 
The novelties of our paper with respect to the recent literature can be summarized as:
\begin{itemize}
\item The equilibrium functions are not forced to be (piecewise) linear and are not subject to slope limitations~\cite{GC-RC:17,KT-PS-SB-MC:11,HZ-HJL:15}. This relaxes several restrictive constraints on equilibrium functions, which
leads to an enlarged search space of potential candidates of desired OPRF surrogates, thus reducing the optimality gap.
\item The control rule is globally asymptotically stable, as opposed to algorithms that are locally asymptotically stable~\cite{XZ-MF-ZL-LC-SHL:21,WC-JL-BZ:22} or whose stability is not analytically characterized.
Our design provably steers the system to the desirable configuration described by the equilibrium functions irrespective of the initial reactive power injection.
\item We introduce the idea of pseudo data points to enhance the voltage regulation capability of learned controllers when voltages are not within the desired limits.
\end{itemize}
With respect to its preliminary version~\cite{GC-ZY-MKS-JC:22-cdc}, this work differs as follows: it relaxes the differentiability requirement on the control rule;
it establishes global, rather than local, stability guarantees for the control scheme; 
it provides a more general design of neural networks in the learning process; 
and it uses of pseudo data points.

\subsubsection*{Outline}
The paper is organized as follows.
In Section~\ref{sec:modeling}, we model a power distribution network and define the problem of interest.
Section~\ref{sec:control_rule} introduces the control scheme, states its stability properties, and identifies conditions on the equilibrium functions needed for the system stability. Sections~\ref{sec:learn_eq_function} describes the learning process: precisely, it reports how we construct the data set used for learning and how we parameterize the equilibrium function with a neural network to meet the required conditions by design. 
Numerical simulations in Section~\ref{sec:tests} validate the proposed approach and show significant improvements with respect to prevalent benchmark methods. Finally, Section~\ref{sec:conc} concludes this work.


\subsubsection*{Notation}
Throughout the paper, $\real$ and $\complex$ denote the set of real and complex numbers, respectively. Upper and lowercase boldface letters denote matrices and column vectors, respectively. Sets are represented by calligraphic symbols. Given a vector $\abf$ (a diagonal matrix $\mathbf A$), its $n$-th (diagonal) entry is denoted by $a_n$ $(A_n)$. $\Abf \succ(\succeq)~0$ denotes that matrix $\Abf$ is positive (semi-) definite, and $\Abf \prec(\preceq)~0$ denotes that matrix $\Abf$ is negative (semi-) definite. The symbol $(\cdot)^\top$ stands for transposition, and $\ones,\zeros, \Ibf$ denote vectors of all ones and zeros and identity matrix with appropriate dimensions, respectively. Operators $\Re(\cdot)$ and $\Im(\cdot)$ extract the real and imaginary parts of a complex-valued argument, and act element-wise. With a slight abuse of notation, we use $|\cdot|$ to denote the absolute value for real-valued arguments, the magnitude for complex-valued arguments, and the cardinality when the argument is a set. $\|\cdot\|$ represents the Euclidean norm. Given a symmetric matrix $\Abf$, $\lambda_{\max}(\Abf)$ and $\lambda_{\min}(\Abf)$ represent its largest and smallest eigenvalue, respectively. For any matrix $\Bbf$, it holds that $\|\Bbf\| = \sqrt{\lambda_{\max}(\Bbf^\top\Bbf)}$. The graph of a function $\phi: \real \rightarrow \real$ is the set of all points of the form $(x, \phi(x))$, whereas the range of $\phi$ is the set of its possible output values.


\section{Grid Modeling and Problem Formulation}\label{sec:modeling}

Consider a balanced three-phase power distribution network\footnote{Although the analysis and algorithms developed in this work consider balanced grids, contemporary research has demonstrated applicability of related approaches to unbalanced multiphase networks~\cite{SG-SC-VK:22,JF-YS-GQ-SHL-AA-AW:22}. We leave for future research a rigorous and detailed extension of our findings to encompass unbalanced grids.} with $N+1$ buses represented by its single-phase equivalent and modeled as an undirected graph $\Gc = (\Nc, \Ec)$, where $\Nc = \{0, 1, \dots, N\}$ are associated with the electrical buses, and $\Ec$ represents the set of the electrical lines between these buses. We label the substation node as 0, and assume that it behaves as an ideal voltage source imposing the nominal voltage of 1 p.u.
Define the following quantities:
\begin{itemize}
\item $u_n\in \complex$ is the voltage phasor at bus $n\in \Nc$;
\item $v_n\in \real$ is the voltage magnitude at bus $n\in \Nc$;
\item $i_n\in \complex$ is the injected current phasor at bus $n\in \Nc$;
\item $s_n=p_n+iq_n \in \complex$ is the nodal complex power injection at bus $n\in \Nc$, where $p_n,q_n\in \real$ are the active and reactive powers, respectively. Powers take positive (negative) values, i.e., $p_n, q_n \geq 0$ ($p_n, q_n \leq 0$), when they are \emph{injected into} (\emph{absorbed from}) the grid.
\end{itemize}
We use vectors $\ubf,\ibf,\sbf \in \complex^N$ to collect the complex voltages, currents, and complex powers of buses $1,2,\ldots,N$, and vectors $ \vbf,\pbf, \qbf \in \real^N$ to collect their voltage magnitudes, and active and reactive power injections.
Denote by $z_e \in \complex$ and by $y_e = z_e^{-1} \in \complex$ the impedance and the admittance of line $e = (m,n) \in \Ec$, respectively.
The network bus admittance matrix $\Ybf \in \complex^{(N+1)\times(N+1)}$ is a symmetric matrix that can be expressed as $\Ybf = \Ybf_L + \diag(\ybf_T)$, where  
\begin{align*}
(\Ybf_L)_{mn} = \begin{cases}
- y_{(m,n)} & \text{ if } (m,n) \in \Ec, m \neq n, \\
0 & \text{ if } (m,n) \notin \Ec, m \neq n, \\
\sum_{k \neq n} y_{(k,n)} & \text{ if }m = n,
\end{cases}
\end{align*}
and the vector $\ybf_T$ collects the shunt components of each bus. The matrix $\Ybf_L$ is a complex Laplacian matrix, and hence satisfies $ \Ybf_L \ones = \zeros$.
We partition the bus admittance matrix by separating the components associated with the substation and the ones associated with the other nodes, obtaining
\begin{align*}
\Ybf = \begin{bmatrix}
 y_{0}&\ybf_0^\top \\
 \ybf_0& \tilde \Ybf
 \end{bmatrix} ,
\end{align*}
with $y_{0} \in \complex$, $\ybf_0 \in \complex^{N}$, and $\tilde \Ybf \in \complex^{N \times N}$.
If the network is connected, then $\tilde \Ybf$ is invertible~\cite{AMK-MP:17}. Let $\tilde \Zbf := \tilde \Ybf^{-1}$, the power flow equations are
\begin{subequations}
\label{eq:PFeq}
\begin{align}
  \ubf &= \tilde \Zbf \ibf+ \hat \ubf , \label{eq:nodevoltage}\\
  u_0 &= 1, \label{eq:PCCidealvoltgen}\\
  i_0 &= \ones^\top \ibf, \label{eq:PCCidealcurrent}\\
 u_n \bar i_n & = p_n + j q_n, \qquad n\neq 0, \label{eq:nodeconstpwr}\\
 v_n & = |u_n| \label{eq:vm_def}
\end{align}    
\end{subequations}
where $\bar i_n$ denotes the complex conjugate of $i_n$ and $\hat \ubf := \tilde \Zbf \ybf_0$. Eq.~\eqref{eq:nodevoltage} represents the Kirchoff equations and provides the relation between voltages and currents. Eqs.~\eqref{eq:PCCidealvoltgen} and~\eqref{eq:PCCidealcurrent} hold because the substation is modeled as the slack bus.  Eq.~\eqref{eq:nodeconstpwr} comes from the fact that all the nodes, except the substation, are modeled to be constant power buses.

Assume a subset $\Cc \subseteq \Nc$ of buses hosts DERs, with $|\Cc| = C$.
Every DER corresponds to a smart agent provided with some computational sensing capability, i.e., it measures its voltage magnitude.
The remaining nodes constitute the set $\Lc=\Nc\setminus\Cc$ and are referred to as loads.
For convenience, we partition the reactive powers and the voltage magnitudes by grouping together the nodes belonging to the load and generation sets
\begin{align*}
\qbf = \begin{bmatrix}
\qbf_\Cc^\top & \qbf_\Lc^\top
\end{bmatrix}^\top, \quad
\vbf = \begin{bmatrix}
\vbf_\Cc^\top & \vbf_\Lc^\top
\end{bmatrix}^\top.
\end{align*}

Motivated by practical considerations, we assume that the grid is not endowed with a communication network that can be used by agents, loads, and the system operator to share information in \emph{real-time}. As a notable consequence, load demands are not known to the system operator or to agents in real time, preventing the use of open-loop strategies for computing the power outputs. The (averaged) demands could, however, be reported on a hourly/daily basis for slow-timescale applications such as billing.

The massive deployment of DERs in DGs might induce voltage quality issues. For example, sudden generation drops could lead the voltages of a network with high penetration of renewables below desired operational limits and even close to collapse. Since DERs are able to provide ancillary services, reactive power compensation can be used to regulate voltage profiles. Ideally, one wants the DER reactive power setpoints to be the solution of an \emph{optimal reactive power flow} (ORPF) problem of the form\footnote{More comprehensive OPRF problems could in principle be of interest in practical applications, e.g., considering line flows limitations as well. Although in this paper we focus on OPRF problems of the type~\eqref{eq:ORPF}, our approach can be readily applied to other ORPF formulations. Also, we restrict our attention to the case in which problem~\eqref{eq:ORPF} admits a unique solution. When that is not the case, $\qbf_\Cc^\star(\pbf,\qbf_\Lc)$ can be chosen among the set of minimizers.}
\begin{subequations}\label{eq:ORPF}
	\begin{align}
	\qbf_\Cc^\star(\pbf,\qbf_\Lc):=\arg\min_{\qbf_\Cc}\ &  ~f(\qbf_\Cc) \label{eq:ORPF:cost}\\
	\mathrm{s.t.}\  & ~\eqref{eq:nodevoltage}-\eqref{eq:vm_def}\notag\\
	&~\vbf_{\min} \leq \vbf(\qbf_\Cc) \leq \vbf_{\max} \label{eq:ORPF:c1}\\
    &~\qbf_{\min} \leq \qbf_\Cc \leq \qbf_{\max} \label{eq:ORPF:c2}
	\end{align} 
\end{subequations}
where $\qbf_{\min}, \qbf_{\max} \in \real^C$ are the minimum and maximum DERs' reactive power injections; $\vbf_{\min}, \vbf_{\max} \in \real^N$ are the desired voltage lower and upper bounds on \emph{all} the network buses; and $f:\real^C \rightarrow \real$ is the cost function of interest.
The minimizer depends on the uncontrolled variables $\pbf$ and $\qbf_\Lc$, which appear implicitly in the constraint~\eqref{eq:ORPF:c1} via equation~\eqref{eq:vm_def}.
Also, notice that for given (re)active loads and active generation, the voltage at node $n$ and the vector of the voltage magnitudes become functions exclusively of $\qbf_\Cc$, i.e., 
$$v_n = v_n(\qbf_\Cc), \quad \vbf = \vbf(\qbf_\Cc).$$
The complexity of solving~\eqref{eq:ORPF} depends on the choice of the cost function $f(\cdot)$ and the non-convexity of~\eqref{eq:PFeq}. Tremendous advancements have been made to overcome the computational limitations via convex relaxations~\cite{SHL:14}, linearized power flow equations~\cite{ED-AS:18}, distributed optimization~\cite{ED-HZ-GBG:13}, and learning-based approaches~\cite{MKS-SG-VK-GC-AB:20}. However, solving~\eqref{eq:ORPF} inevitably requires the knowledge of network-wide quantities $(\pbf,\qbf_\Lc)$, centrally or via peer-to-peer communication. Because the necessary supporting real-time communication infrastructure is not prevalent for most distribution systems, the optimal $\qbf_\Cc^\star$ cannot be directly computed. 
Rather, inspired by the ongoing efforts towards designing local communication-free control rules for DERs and the recently reported success of neural-network-based surrogates for OPF, this work proposes a two-stage approach.
In the first stage, termed the \emph{learning stage}, historical data are used to learn functions that map voltages to (approximate) solutions of the ORPF problem~\eqref{eq:ORPF}. Specifically, for each agent $n \in \mathcal C$, we learn a function $\phi_n$
$$\phi_n: \real \rightarrow \real, \; v_n \mapsto \phi_n(v_n)$$
that takes as an input the local voltage $v_n$, and provides as an output the approximated ORPF solution. Given any voltage $v_n$, $\phi_n(v_n)$ represents an approximation of the reactive power that the DER at node $n$ would inject if its voltage is $v_n$ and the network is operating at a solution of~\eqref{eq:ORPF}.
The graph of $\phi_n$, namely, points of the form $(v_n,\phi_n(v_n))$, consists then of the ORPF solutions' surrogates and describes desirable network configurations.
The second stage, termed the \emph{control stage}, aims to design local control rules that steer the network to the aforesaid desirable configurations. That is, the controller equilibrium points are determined by the $\phi_n$'s and, for 
this reason, the $\phi_n$'s are hereafter called \emph{equilibrium functions}.
We introduce first in Section~\ref{sec:control_rule} the local control scheme and derive conditions ensuring system stability. We build on these conditions in Section~\ref{sec:learn_eq_function} to guide the learning of $\{\phi_n\}_{n \in \Cc}$ to promote the efficient operation of the DG. 


\section{Provably Stable Local Volt/Var Control}\label{sec:control_rule}

Here, we propose a local Volt/Var controller and analyze its stability properties. For each $n \in \Cc$, given the equilibrium function $\phi_n$, consider the reactive power update
\begin{align}\label{eq:bus_react_upd} 
    q_n(t+1) = q_n(t) + \epsilon (\phi_n(v_n(t)) - q_n(t)) .
\end{align} 
In~\eqref{eq:bus_react_upd}, the new reactive power setpoint is chosen as a convex combination between the previous one and the equilibrium function evaluated at the current voltage.
Rules such as~\eqref{eq:bus_react_upd} are referred to as \emph{incremental}, because the updated reactive power setpoint is obtained by adding to the previous one an increment weighted by the stepsize parameter $\epsilon \in [0,1]$.
An equilibrium point of~\eqref{eq:bus_react_upd}, denoted $\qbf_\Cc^\sharp$, satisfies for $n\in \mathcal C$
\begin{subequations}
\label{eq:fixed_point}
\begin{align}
& q_n^\sharp = \phi_n(v_n)\\
& v_n^\sharp = v_n(\qbf_\Cc^\sharp). 
\end{align}
\end{subequations}
That is, $(\phi_n(v_n^\sharp), v_n^\sharp)$ belongs to the graph of $\phi_n$ and hence is a desirable configuration.
It remains now to establish under what conditions the algorithm~\eqref{eq:bus_react_upd} converges to an equilibrium.
The following convergence analysis assumes that uncontrolled variables, namely, $\pbf$ and $\qbf_\Lc$, assume arbitrary, but fixed in time, values. 

As customary in the literature of reactive power control for DGs, e.g., see~\cite{ED-AS:18,GQ-NL:20,HZ-HJL:15}, we rely on the following linearization of the power flow equations to study the stability properties of~\eqref{eq:bus_react_upd}.
In general, voltage magnitudes are nonlinear functions of the nodal power injections. Define $\tilde \Rbf:=\Re(\tilde \Zbf)$ and $\tilde \Xbf:=\Im(\tilde \Zbf) \in \real^{N \times N}$, by using a first-order Taylor expansion, the power flow equation can be linearized to obtain~\cite{GC-RC:17}
 \begin{align}
 \vbf = \tilde \Rbf \pbf + \tilde \Xbf \qbf + |\hat \ubf|.
 \label{eq:v=Rp+Xq}
 \end{align}
Also, the matrices $\tilde \Rbf$ and $\tilde \Xbf$ can be decomposed according to the former partition, yielding
\begin{align*}
\tilde \Rbf = \begin{bmatrix}
\Rbf & \Rbf_\Lc \\
\Rbf_\Lc^\top & \Rbf_{\Lc\Lc}
\end{bmatrix}, \quad
\tilde \Xbf = \begin{bmatrix}
\Xbf & \Xbf_\Lc \\
\Xbf_\Lc^\top & \Xbf_{\Lc\Lc}
\end{bmatrix},
\end{align*}
with $\Rbf, \Xbf \succ 0$~\cite{HZ-HJL:15}.
From \eqref{eq:v=Rp+Xq}, voltage magnitudes become functions exclusively of $\qbf_\Cc$:
 \begin{align}\label{eq:volt-approx}
 &\vbf(\qbf_\Cc) = 
 \begin{bmatrix}
 \Xbf \\
 \Xbf_\Lc^\top
 \end{bmatrix} \qbf_\Cc + \hat \vbf,
\end{align}
where
\begin{align}\label{eq:hatv-def}
\hat \vbf &:= \begin{bmatrix}
\hat \vbf_\Cc\\
\hat \vbf_\Lc
\end{bmatrix} = 
\begin{bmatrix}
\Xbf_\Lc\\
\Xbf_{\Lc\Lc}
\end{bmatrix} \qbf_\Lc + \tilde\Rbf \pbf + |\hat \ubf|.
\end{align}
Collecting the $\{\phi_n\}_{n \in \Cc}$ in the vector-valued function
$\phib$, and  adopting the linearization~\eqref{eq:volt-approx}, the system dynamics can be described in compact form as
\begin{subequations}\label{eq:dyn_sys}
\begin{align} 
    & \qbf_\Cc(t+1) = \qbf_\Cc(t) + \epsilon (\phib(\vbf_\Cc(t)) - \qbf_\Cc(t)), \label{eq:dyn_sys_q}\\
    & \vbf_\Cc(t+1) = \Xbf \qbf_\Cc(t+1) + \hat \vbf_\Cc, \label{eq:dyn_sys_v}
\end{align}
\end{subequations}
where it is tacitly assumed that, at the timescale of the above iterates, the exogenous variables $(\pbf,\qbf_\Lc)$ remain constant, resulting in a constant term $\hat \vbf_\Cc$ from~\eqref{eq:hatv-def}. Let  $(\qbf_\Cc^\sharp,\vbf_\Cc^\sharp)$ an equilibrium point of~\eqref{eq:dyn_sys}. By definition, it must satisfy 
\begin{subequations}\label{eq:fixed-point}
\begin{align}
    \qbf_\Cc^\sharp &= \phib(\vbf_\Cc^\sharp), \label{eq:fixed-point-q} \\ 
    \vbf_\Cc^\sharp &= \Xbf \qbf_\Cc^\sharp + \hat \vbf_\Cc. \label{eq:fixed-point-v}
\end{align}
\end{subequations}
In particular, it holds that $q_n^\sharp = \phi_n(v_n^\sharp)$ for each $n \in \Cc$, see~\eqref{eq:fixed_point}, showing that the functions $\{\phi_n\}_{n \in \Cc}$ describe all the possible equilibrium points of~\eqref{eq:dyn_sys_q}.

Define the feasible set of reactive power injections $\Qc := \times_{n \in \Cc} \Qc_n$, with $\Qc_n = \{ q_n : q_{\min,n} \leq q_n \leq q_{\max,n} \}$.
To ensure the stability of the proposed local algorithm, we require each $\phi_n$ to meet the following properties:
\begin{enumerate}[leftmargin=3em]
\item[C1)] $\phi_n$ is Lipschitz, i.e., there exists $L_n < \infty$ such that $|\phi_n(v)-\phi_n(v')| \leq L_n|v-v'|$, for all $v, v' \in \real$; 
\item[C2)] $\phi_n$ is nonincreasing in $v_n$;
\item[C3)] $\range(\phi_n) \subseteq \Qc_n$, i.e., $\phi_n: \real \rightarrow \Qc_n$. 
\end{enumerate}
The role of conditions C1) -- C3) will be clear from the following formal results proved in the Appendix. The first one characterizes the equilibrium points of~\eqref{eq:dyn_sys}. 

\begin{proposition}\longthmtitle{Feasibility of the reactive power update and uniqueness of the  equilibrium}\label{prop:eq-point}
Let $\{\phi_n\}_{n \in \Cc}$ satisfy the conditions C1) -- C3), and assume $\qbf_C(0) \in \Qc$. The reactive power update~\eqref{eq:bus_react_upd} is feasible, i.e., $\qbf_C(t) \in \Qc$, $t \geq 0$.
Moreover, the system~\eqref{eq:dyn_sys} has an unique equilibrium point $(\qbf_\Cc^\sharp,\vbf_\Cc^\sharp)$.
\end{proposition}

The next result characterizes the stability properties of the equilibrium point identified in Proposition~\ref{prop:eq-point}.

\begin{proposition}\longthmtitle{Global asymptotic stability of the control rule~\eqref{eq:bus_react_upd}}\label{prop:convergence}
Let $\{\phi_n\}_{n \in \Cc}$ satisfy the conditions C1) -- C3), and assume $\qbf_C(0) \in \Qc$. Define $L = \max_{n \in \Cc} L_n$. If $\epsilon$ is such that 
\begin{equation}
    0 < \epsilon < \min \Big\{ 1,\frac{2}{(\|\Xbf\|L + 1)^2} \Big\}
    \label{eq:conv_cond}
\end{equation}
then the equilibrium of~\eqref{eq:dyn_sys} is \emph{globally asymptotically stable}.
\end{proposition}
Proposition~\ref{prop:convergence} indicates that, as long as the $\phi_n$'s meet the conditions C1) -- C3), one can always find $\epsilon>0$ so that $(\qbf_\Cc,\vbf_\Cc)$ converges to the unique equilibrium point $(\qbf_\Cc^\sharp,\vbf_\Cc^\sharp)$ under the reactive power update rule~\eqref{eq:bus_react_upd}.
Note that, in condition~\eqref{eq:conv_cond}, $\|\Xbf\|$ is fully determined by the DG and $L$ can be computed once the functions $\{\phi_n\}_{n \in \Cc}$ have been selected.
Because $\qbf_\Lc$ and $\pbf$ are fixed, the convergence of $\qbf_\Cc$ leads, cf.~\eqref{eq:volt-approx}, also to the global asymptotic convergence of $\vbf$.
Finally, we note that the proposed reactive power update rule~\eqref{eq:bus_react_upd} is a generalized version of the local control scheme proposed in~\cite{GC-RC:17} which considers only linear functions (instead of arbitrary nonlinear $\{\phi_n\}_{n \in \Cc}$ satisfying  C1) -- C3)).

\begin{remark}\longthmtitle{Global vs. local asymptotic stability of the equilibrium}
In our previous work \cite{GC-ZY-MKS-JC:22-cdc}, the equilibrium point of~\eqref{eq:dyn_sys} under the reactive power update rule~\eqref{eq:bus_react_upd} is shown to be locally asymptotically stable if $0 < \epsilon < \min \{1,\frac{2}{\|\Xbf\|L + 1} \}$.
The previous claim roughly implies that if $\qbf_\Cc(0)$ is close enough to $\qbf_\Cc^\sharp$, then it converges to $\qbf_\Cc^\sharp$. 
Our result here extends the stability properties from local to global, at the cost of reducing the selection range of $\epsilon$, as $\frac{2}{\|\Xbf\|L + 1} > \frac{2}{(\|\Xbf\|L + 1)^2}$. \oprocend
\end{remark}

\begin{remark}\longthmtitle{Non-incremental vs. incremental reactive power update rules}\label{rmk:slope-constraint}
Many works in the literature consider local Volt/Var control schemes of the form~\cite{IEEE1547,KT-PS-SB-MC:11,HZ-HJL:15}
\begin{equation}
\label{eq:non_incr}
  q_n(t+1) = \varphi_n(v_n(t)) .
\end{equation}
Reactive power update rules such as~\eqref{eq:non_incr}, where the new reactive power setpoints are determined based on the local voltage without explicitly exploiting a memory of past setpoints, are referred to as \emph{non-incremental}~\cite{MF-XZ-LC:15}. 
The equilibrium points of~\eqref{eq:non_incr} satisfy
\begin{align*}
 q_n &= \varphi_n(v_n) , \\
 v_n &= v_n(\qbf_\Cc) ,
\end{align*}
i.e., $\varphi_n(v_n)$ plays the double role of the control function and the  equilibrium function. Thus, even the control rule
\begin{align}\label{eq:ni_bus_react_upd}
    q_n(t+1) &= \phi_n(v_n(t)). 
\end{align}
looks appealing for our framework because its equilibria are determined by $\phi_n$, this rule corresponds to $\epsilon = 1$ in~\eqref{eq:bus_react_upd}.
From the proof of Proposition~\ref{prop:convergence}, one can show that the algorithm~\eqref{eq:ni_bus_react_upd} is globally asymptotically stable if the equilibrium functions meet not only  C1) -- C3) but also
\begin{align}
\label{eq:static_global_conv}
\|\Xbf\|L < \sqrt{2}-1.    
\end{align}
The former condition bounds the slope of the $\{\phi_n\}_{n \in \Cc}$ and appears often in the literature~\cite{GC-ZY-MKS-JC:22-cdc,XZ-MF-ZL-LC-SHL:21,GC-RC:17}. This means that, to ensure the stability of~\eqref{eq:non_incr}, we need to restrict the search space of potential candidates of $\{\phi_n\}_{n \in \Cc}$, thus risking a degradation in system performance in terms of the optimality gap at the equilibrium.
The aforementioned restriction motivates adopting an incremental control like~\eqref{eq:bus_react_upd}.
\oprocend
\end{remark}


\section{Learning Equilibrium Functions For Efficient Network Operation}\label{sec:learn_eq_function}

Having established the conditions on equilibrium functions for system stability, here we lay out a data-driven approach to synthesize the functions
$\{\phi_n\}_{n \in \Cc}$ to improve operational efficiency of the DG.
Specifically, our goal is to learn local equilibrium functions $\{\phi_n\}_{n \in \Cc}$ under which the system equilibrium $\qbf_\Cc^\sharp(\pbf,\qbf_\Lc)$ is as close as possible to the ORPF problem solution $\qbf_\Cc^\star(\pbf,\qbf_\Lc)$. 
The learning process consists of the following steps.
First, given that the solution of~\eqref{eq:ORPF} depends on $(\pbf,\qbf_\Lc)$, 
we build a set $\{(\pbf^k,\qbf_\Lc^k)\}_{k=1}^K$ of $K$ load-generation scenarios. One can  obtain the aforementioned scenarios via random sampling from assumed probability distributions, historical data, or from forecasted conditions for a look-ahead period. 
Second, we solve the ORPF problem~\eqref{eq:ORPF} for these $K$ scenarios to obtain a labeled data set of corresponding minimizers $\Dc = \{(\qbf_{\Cc,k}^\star,\vbf_{\Cc,k}^\star)\}_{k=1}^{K}$, where the parametric dependencies are omitted for notational ease. 
Third, the entries for these minimizers are then separated for each $n \in \Cc$ to obtain data sets of the form $\Dc_n=\{(v_{n,k}^\star,q_{n,k}^\star)\}_{k=1}^K$, and each equilibrium function $\phi_n$ is then trained by solving 
\begin{align}\label{eq:learning_problem}
    \min_{\phi_n}\ & ~\sum_{k=1}^K |q_{n,k}^\star - \phi_n(v_{n,k}^\star) |^2  \\
    \mathrm{s.t.}\  & ~\phi_n~\text{meets the conditions C1) -- C3)}. \notag
\end{align}

Typical approaches to solve~\eqref{eq:learning_problem} involves restricting the search space for function $\phi_n$ via convenient parameterization leading to approaches such as polynomial regression and neural network approximation methods. Here, we adopt the latter. Enforcing the properties C1) -- C3) is, in general, not trivial and depends on many aspects, e.g., the number of considered layers and the used activation functions. In particular, designing monotone neural networks might require additional considerations. Approaches to do so include, for instance, structure-based~\cite{HD-MV:10}, gradient-constrained~\cite{AW-GL:19}, and verification-based~\cite{XL-XH-NZ-QL:20} methods. 
In the following, we provide a single-hidden-layer neural network design framework that achieves  C1) -- C3) and uses the Rectified Linear Unit (ReLU) activation function $${\rm ReLU}(x) = \max(0,x).$$
Note however that, in principle, any continuous and monotonic activation function could be used in our framework, e.g., the Sigmoid or the Tanh\footnote{${\rm Sigmoid}(x) = \frac{1}{1 + e^{-x}},~{\rm Tanh}(x) = \frac{e^{x} - e^{-x}}{e^{x}+e^{-x}}$.}.

Each equilibrium function $\phi_n$ is then modeled using a single-hidden-layer neural network $\mathsf{N}(x)$ as
\begin{align}\label{eq:enc_c3}
\phi_n(x)  =  q_{\max,n} -  {\rm ReLU}(& q_{\max,n} - \mathsf{N}(x)) \notag \\
&+ {\rm ReLU}(q_{\min,n} - \mathsf{N}(x)),
\end{align}
where 
\begin{align}\label{eq:single-layer-NN}
\mathsf{N}(x) = \sum_{h=1}^{H} w_h {\rm ReLU}(x - b_h) + \beta,
\end{align}
with $w_h$ and $b_h$ being the weight and bias of the $h$-th neuron unit, respectively; $\beta$ being an additional bias term applied in the output layer; and $H$ is the number of neuron units in the hidden layer.
Note that conditions C1) and C3) for $\phi_n$ are automatically met because of the Lipschitzness of the ReLU function and because the output of the neural network $\phi_n$ is constrained to the set $\Qc_n$, cf.~\eqref{eq:enc_c3}. Condition C2) is instead encoded by the next result.

\begin{figure*}
\centering
    \includegraphics[scale=0.33]{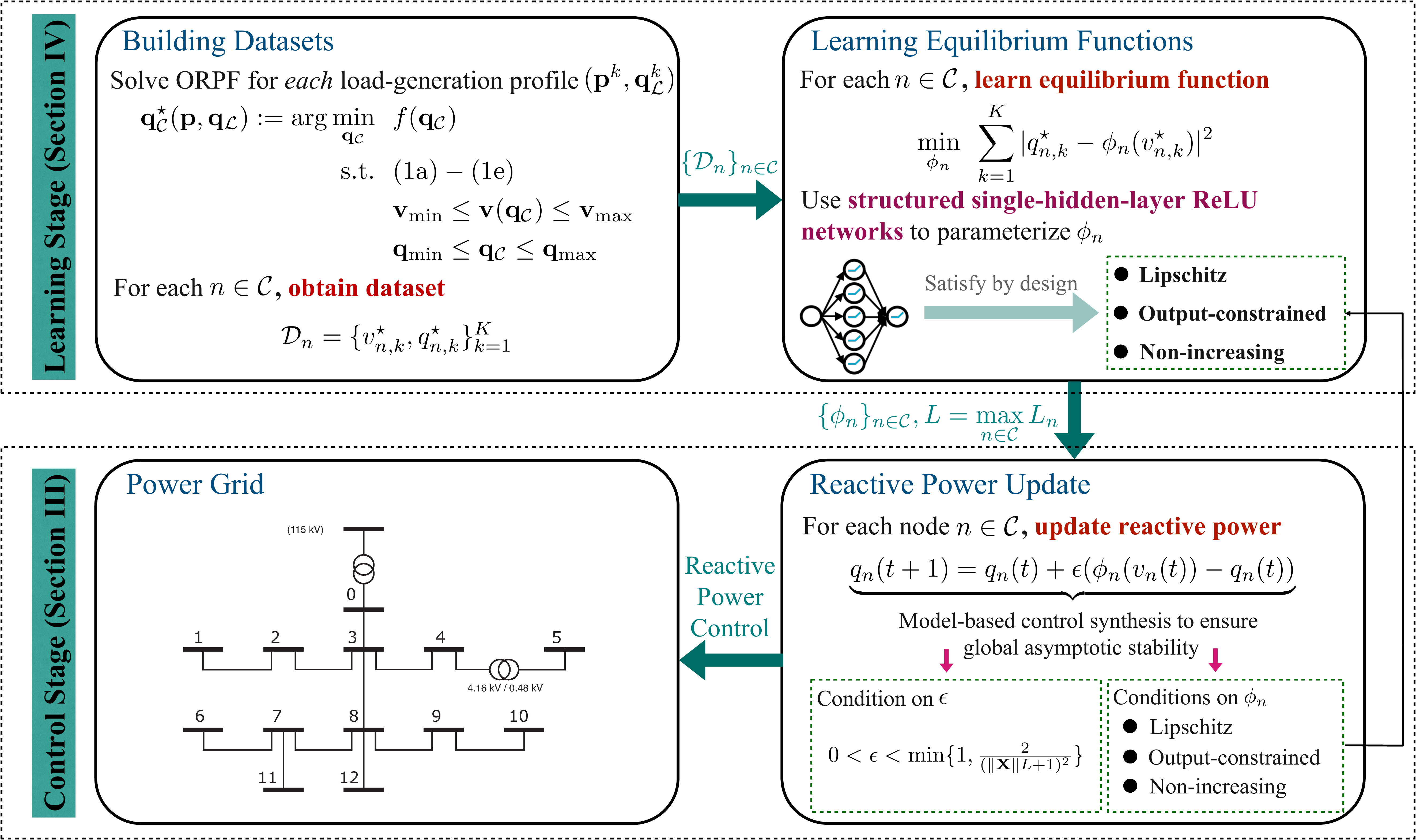}
    \caption{Flowchart of the proposed data-driven control framework.}
    \label{fig:flowchart}
\end{figure*}

\begin{proposition}\longthmtitle{Universal approximation of Lipschitz nonincreasing function using ReLU activation function}\label{prop:universaL-approx}
Consider the single-hidden-layer neural network~\eqref{eq:single-layer-NN}, and reorder the neuron units such that $b_1 \leq b_2 \leq \cdots \leq b_H$.
$\mathsf{N}$ is non-increasing if and only if
\begin{align}\label{eq:enc_c2}
    \sum_{j=1}^{J} w_j \leq 0, \quad \forall J \in \{1,2,\dots,H\}.
\end{align}
Furthermore, for any Lipschitz non-increasing function $g:\real \rightarrow \real$ and given any compact domain $\Xc \in \real$ and $\eta > 0$, there exist $H$, $w_h$, $b_h$, and $\beta$ such that $|\mathsf{N}(x) - g(x)| \leq \eta$ for all $x \in \Xc$.
\end{proposition}
Its proof can be found in the Appendix, and note that the monotonicity of $\mathsf{N}$ remains after the projection~\eqref{eq:enc_c3}.
This design imposes weight constraints~\eqref{eq:enc_c2} and output constraints~\eqref{eq:enc_c3} on a single-hidden-layer neural network~\eqref{eq:single-layer-NN} so that it meets the conditions C1) -- C3) by construction. 
Leveraging the universal approximation property, the optimization problem~\eqref{eq:learning_problem} is then equivalent to the parameterized formulation
\begin{align}\label{eq:tractable_learning_problem}
    \min_{\wbf,\bbf,\beta}\ & ~\sum_{k=1}^K |q_{n,k}^\star - \phi_n(v_{n,k}^\star) |^2  \\
    \mathrm{s.t.}\  & \eqref{eq:enc_c3}, ~\eqref{eq:single-layer-NN}, ~\eqref{eq:enc_c2} \notag
\end{align}
which can be solved to local optima using suitable renditions of (stochastic) gradient descent prevalent for neural network training.
Also, exploiting the fact that the ReLU function is used as an activation function, the Lipschitz constant $L_n$ of each $\phi_n$ can be easily computed, see~\eqref{eq:enc_c2}, as
\begin{align*}
    L_n = \max_{J\in\{1,..,H\}} \Big|\sum_{j=1}^{J} w_j \Big|.
\end{align*}

The two stages of the proposed framework are summarized in Algorithm~\ref{alg:framework}. Fig.~\ref{fig:flowchart} provides a flowchart to illustrate the overall scheme.
\begin{algorithm}[ht]
  \caption{Local Volt/Var Control Framework}\label{alg:framework}
  \begin{algorithmic}[1]
  \renewcommand{\algorithmicrequire}{\underline{\texttt{Learning Stage (OffLine)}}}
  \Require
  \renewcommand{\algorithmicrequire}{\textbf{Input:}}
  \Require 
    Historical data in the form of $K$ load-generation scenarios $\{(\pbf^k,\qbf_\Lc^k)\}_{k=1}^K$
  \State Solve the ORPF problem~\eqref{eq:ORPF} for each load-generation scenario, obtain the optimal reactive powers $\{\qbf^{\star,k}_{\Cc}\}_{k=1}^K$ 
  \State Compute the optimal voltages $\{\vbf^{\star,k}\}_{k=1}^K$ by solving the power flow equations~\eqref{eq:PFeq} with the power injections $\{(\pbf^k,\qbf^{k}_{\Lc},\qbf^{\star,k}_{\Cc})\}_{k=1}^K$
  \State Build the data set of optimal reactive powers and voltages $\Dc = \{(\qbf_{\Cc}^{\star,k},\vbf_{\Cc}^{\star,k})\}_{k=1}^{K}$
  \State Separate the entries of these minimizers to obtain, for each $n \in \Cc$, data sets of the form $\Dc_n=\{(v_{n}^{\star,k},q_{n}^{\star,k})\}_{k=1}^K$
  \State Learn the equilibrium functions $\{\phi_n\}_{n \in \Cc}$ by solving~\eqref{eq:tractable_learning_problem}
  \end{algorithmic}
  \begin{algorithmic}[1]
  \renewcommand{\algorithmicrequire}{\underline{\texttt{Control Stage (OnLine)}}}
  \Require
  \renewcommand{\algorithmicrequire}{\textbf{Input:}}
  \Require Learned equilibrium functions $\{\phi_n\}_{n \in \Cc}$, 
   network parameter $\Xbf$, constant $L = \max_{n \in \Cc} L_n$, stepsize parameter $\epsilon$ selected according to~\eqref{eq:conv_cond}
   \renewcommand{\algorithmicrequire}{{Each agent $n \in \Cc$, for $t\geq 0$, cyclically repeats:}}
  \Require
  \State Measuring its local voltage magnitude $v_n(t)$
  \State Updating the reactive power as per~\eqref{eq:bus_react_upd}
  \end{algorithmic}
\end{algorithm}

\begin{remark}\longthmtitle{Enhancing the capability to regulate voltages when they are not within desired limits through pseudo data points}\label{rmk:fic-points}
In the above exposition, the data set points are solutions to the ORPF problem~\eqref{eq:ORPF}, which are subject to the constraint~\eqref{eq:ORPF:c1}. That is, for each $n \in \Cc$, the equilibrium function $\phi_n$ is trained only using data points such that $v_{\min,n} \leq v_{n,k}^\star \leq v_{\max,n}$, i.e., not when the voltages exceed the limits.
Nevertheless, in practical implementation, a DG might experience load-generation scenarios in which~\eqref{eq:ORPF} is infeasible and the voltages do not meet the desired constraints.
Engineering considerations suggest that in such cases the available reactive power capability should be maximally used to alleviate the voltage violations as much as possible. Namely, for each $n \in \Cc$, if $v_n < v_{\min,n}$ ($v_n > v_{\max,n}$), then $q_n = q_{\max,n}$ $(q_n = q_{\min,n})$, see~\cite{GC-RC:17,KT-PS-SB-MC:11}.
To ensure that the learned function $\phi_n$ meets this condition, we can add a certain number of additional pseudo data points to the data set, e.g., $\underline K$ points of the form
$\{(\underline v_{n,k},q_{\max,n})\}_{k=1}^{\underline K}$,
and $\overline K$ points of the form $\{(\overline v_{n,k},q_{\min,n})\}_{k=1}^{\overline K}$, with $\underline v_{n,k} \leq v_{\min,n}, \overline v_{n,k} \geq v_{\max,n}$. These points could be uniformly spaced or randomly sampled. Here we adopt the former method, illustrated in Fig.~\ref{fig:data_points}. 
\oprocend
\end{remark}

\begin{figure}[h]
    \centering
    \subfigure[Without pseudo data points]{
    \includegraphics[scale=0.29]{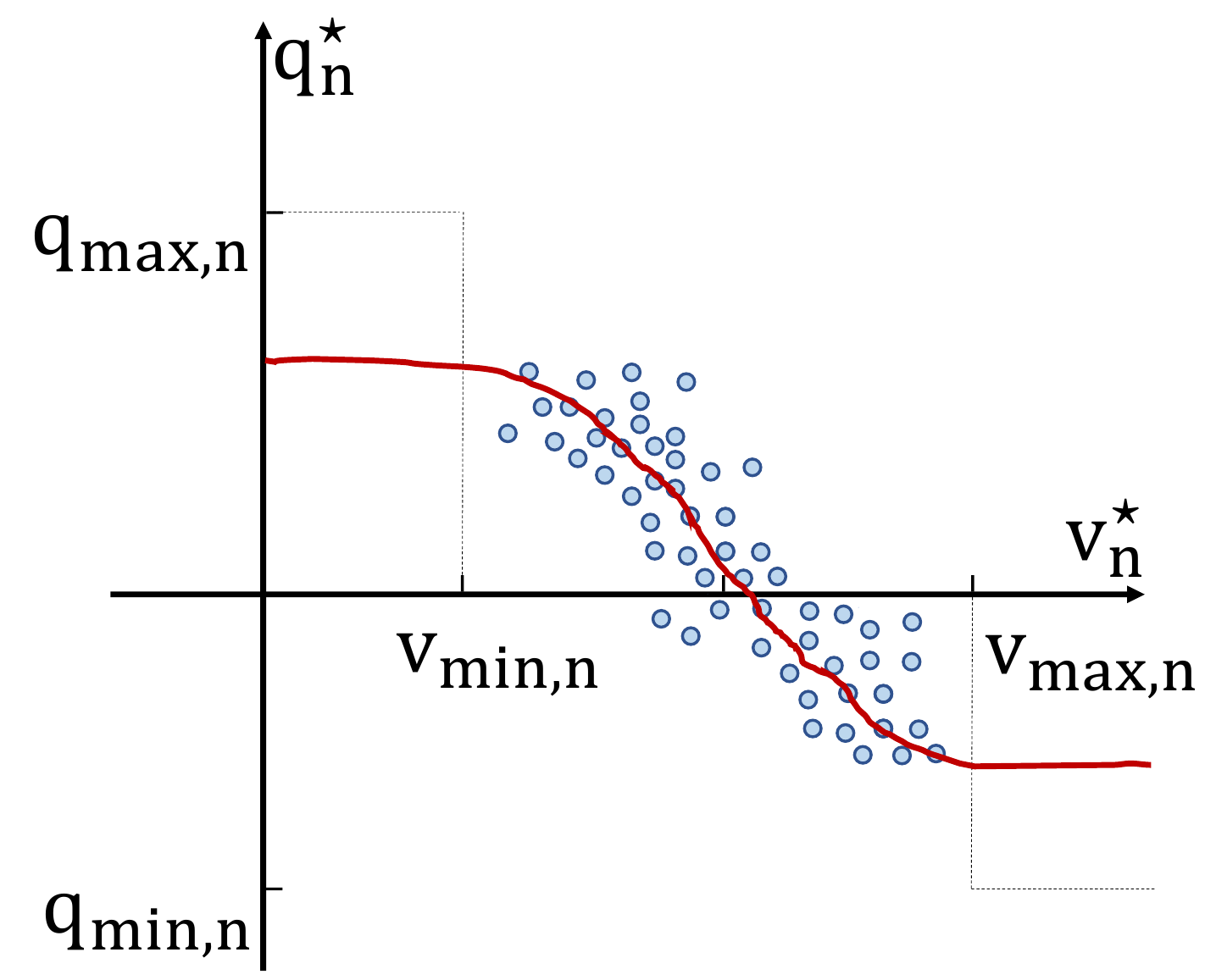}
    }
    \hspace{-.15in}
    \subfigure[With pseudo data points]{
	\includegraphics[scale=0.29]{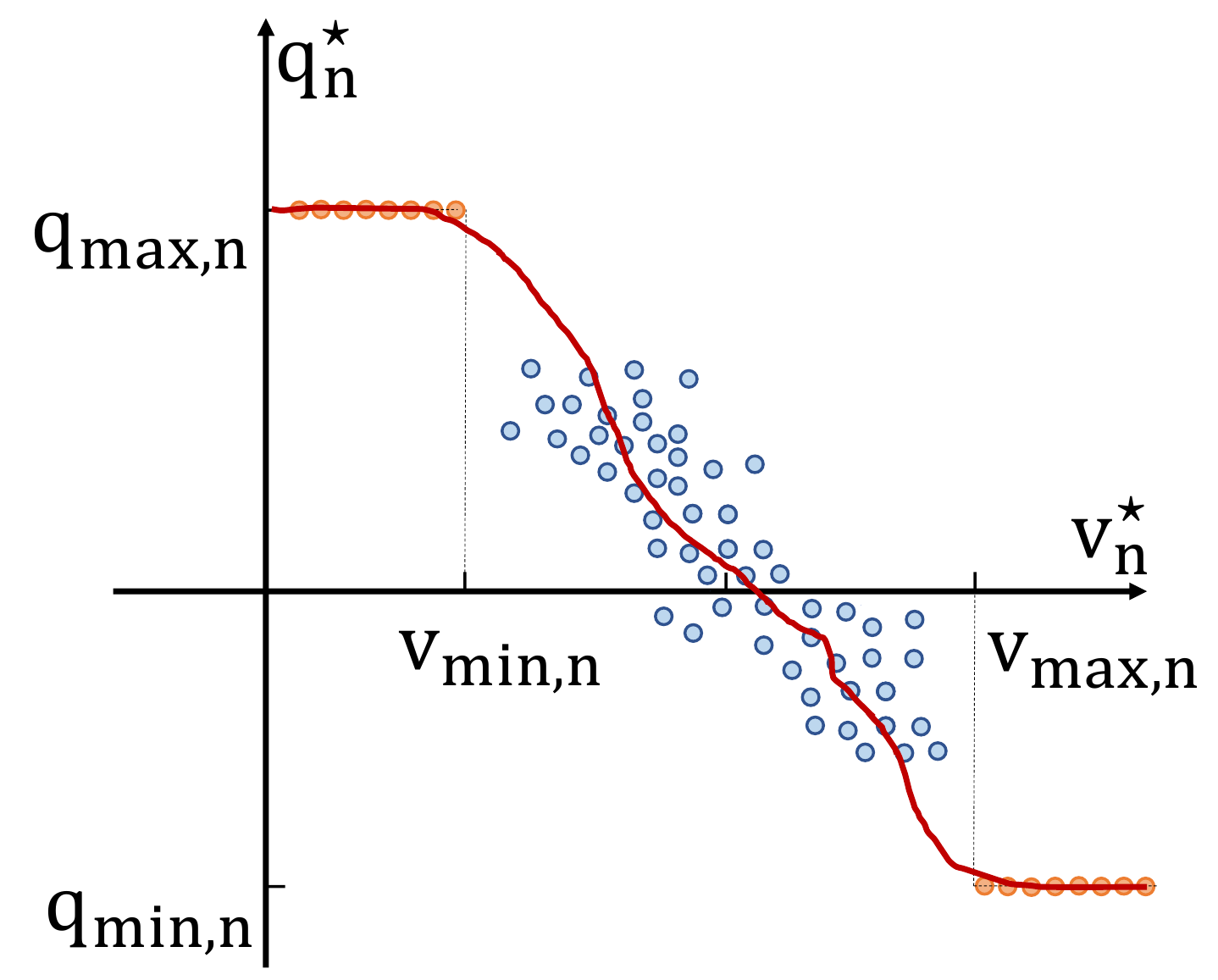}
    }
    \caption{An illustration of the role of the pseudo data points for DER at node $n$. Blue and orange points respectively represent true and pseudo data points, while the dark red curves are instances of learned equilibrium functions. Adding pseudo data points helps the equilibrium functions reach maximum reactive power compensation capability when voltage exceeds the limits.}
    \label{fig:data_points}
\end{figure}

\begin{remark}\longthmtitle{Deep neural network parameterization of non-increasing function via gradient penalization}\label{rmk:DGN}
Though the single-hidden-layer neural network parameterization of~\eqref{eq:single-layer-NN} achieves universal approximation, it requires the ``width'' of the neural networks, i.e., $H$ to be sufficiently large. Instead, in certain situations, deeper neural networks could achieve better approximations than the shallower ones, even if they are much narrower~\cite{IS-OS:17}. On the other hand, the structure-based restrictions enforced on the neural networks to guarantee monotonicity could in some cases 
restrict expressibility and lead to unsatisfactory approximation results~\cite{XL-XH-NZ-QL:20}. To overcome these challenges, we describe a deep neural network parameterization approach by incorporating the monotonicity requirement as a penalization in the cost function of the learning process. Suppose for each $n \in \Cc$, $\phi_n$ is parameterized by a deep neural network, and denote by $d\phi_n(v_n) \in \real$ the (sub)gradient of $\phi_n$ with respect to $v_n$. The cost function in~\eqref{eq:learning_problem} is then replaced by
\begin{align*}
    \sum_{k=1}^K |q_{n,k}^\star - \phi_n(v_{n,k}^\star)|^2 + \gamma_n \max(0,\;d\phi_n(v_{n,k}^\star)),
\end{align*}
where $\gamma_n>0$ is a tuning parameter. During implementation, one can gradually increase $\gamma_n$ until $\phi_n$ is verified to be non-increasing.
\oprocend
\end{remark}

\begin{remark}\longthmtitle{On the comparison with existing reinforcement learning-based approaches}
Recent literature has also investigated reinforcement learning approaches for learning stability guaranteed local Volt/Var controllers, e.g.,~\cite{WC-JL-BZ:22,YS-GQ-SL-AA-AW:22}. However, due to the lack of communication in the training phase, the cost function
that the whole system seeks to minimize in such settings can only be separable, i.e., the summation of all the local cost functions at each node. Therefore, these approaches generally cannot cope with cost functions such as power losses
which shows  coupling among nodes. In contrast, since our approach only uses off-line collected data, any type of cost function could, in principle, be considered when solving the ORPF problem~\eqref{eq:ORPF}.
\oprocend
\end{remark}


\section{Numerical Tests}\label{sec:tests}

\begin{figure}[t]
\centering	
\includegraphics[width=0.7\columnwidth]{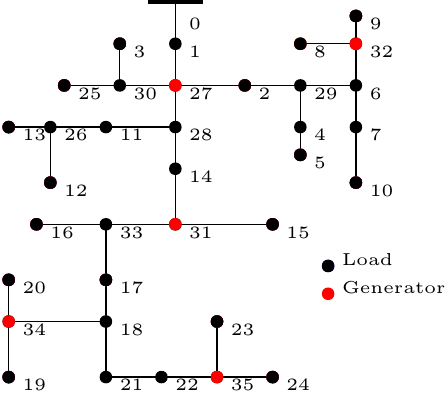}
\caption{The IEEE 37-bus feeder.}
\label{fig:ieee37}
\end{figure}

We conduct case studies on the IEEE 37-bus feeder. We omit regulators, incorporate five solar generators, and convert it to its single-phase equivalent, see Fig.~\ref{fig:ieee37}. The feeder has 25 buses with non-zero load, and the five solar generators are the DERs participating in reactive power compensation.

For our experiments, we use minute-based load and solar generation data retrieved from the Pecan Street data set (June 1, 2018)~\cite{PecanData}. The first 75 nonzero load buses from the data set are aggregated every 3 loads and normalized to obtain 25 load profiles. Similarly, we obtain 5 solar generation profiles for the active power of DERs.
The normalized load profiles for the 24-hour period are scaled so that 
the demand peak is 1.65 times the nominal load.
We synthesize the reactive loads by scaling the active demand to match the power factors of the IEEE 37-bus feeder. Fig.~\ref{fig:solarload} shows the total demand and solar generation across the feeder.

To evaluate the effectiveness of our control framework, it is compared with the standard linear droop control from~\cite{GC-RC:17,KT-PS-SB-MC:11} which dictates
\begin{align*}
&q_n(t+1) = \varrho_n(v_n(t))  \\
&\varrho_n(v_n) := \begin{cases}
q_{\max,n} & v_n(t) \leq v_{\min,n}, \\
q_{\min,n} & v_n(t) \geq v_{\max,n}, \\
-c_n(v_n(t) \!-\! v_{\min,n}) \!+\! q_{\max,n} & \text{otherwise},
\end{cases}
\end{align*}
where $c_n = \frac{q_{\max,n}-q_{\min,n}}{v_{\max,n}-v_{\min,n}}$, and with the  the optimized droop control design\footnote{In~\cite{HJ-CW-PL-JZ-GS-FD-JW:18}, the voltage limits for maximum reactive power provision and absorption are selected as $v_{\min,n} = 0.9$ p.u. and $v_{\min,n} = 1.1$ p.u, respectively. To ensure the DERs reach maximum reactive power compensation capability when voltage exceeds the limits, as discussed in Remark~\ref{rmk:fic-points}, here they are set to $v_{\min,n} = 0.95$ p.u. and $v_{\min,n} = 1.05$.} from~\cite{HJ-CW-PL-JZ-GS-FD-JW:18}
\begin{align*}
&q_n(t+1)= \rho_n(v_n(t)) \\
&
\rho_n(v_n):= \begin{cases}
q_{\max,n} & v_n \leq v_{\min,n}, \\
\frac{\bar v_{\min,n} - v_n}{\bar v_{\min,n} - v_{\min,n}} q_{\max,n} & v_{\min,n} < v_n < \bar v_{\min,n}, \\
0 & \bar v_{\min,n} \leq v_n \leq \bar v_{\max,n}, \\
\frac{v_n - \bar v_{\max,n}}{\bar v_{\max,n} - \bar v_{\max,n}} q_{\min,n} & \bar v_{\max,n} < v_n < v_{\max,n}, \\
q_{\min,n} & v_n \geq v_{\max,n}
\end{cases}
\end{align*}
where $\bar v_{\min,n}$ and $ \bar v_{\max,n}$ are parameters satisfying $v_{\min,n} < \bar v_{\min,n} \leq \bar v_{\max,n} < v_{\max,n}$ which are optimized given the data set prescribing a day-ahead forecast.
\begin{figure}[t]
    \centering
    \includegraphics[width=.52\textwidth]{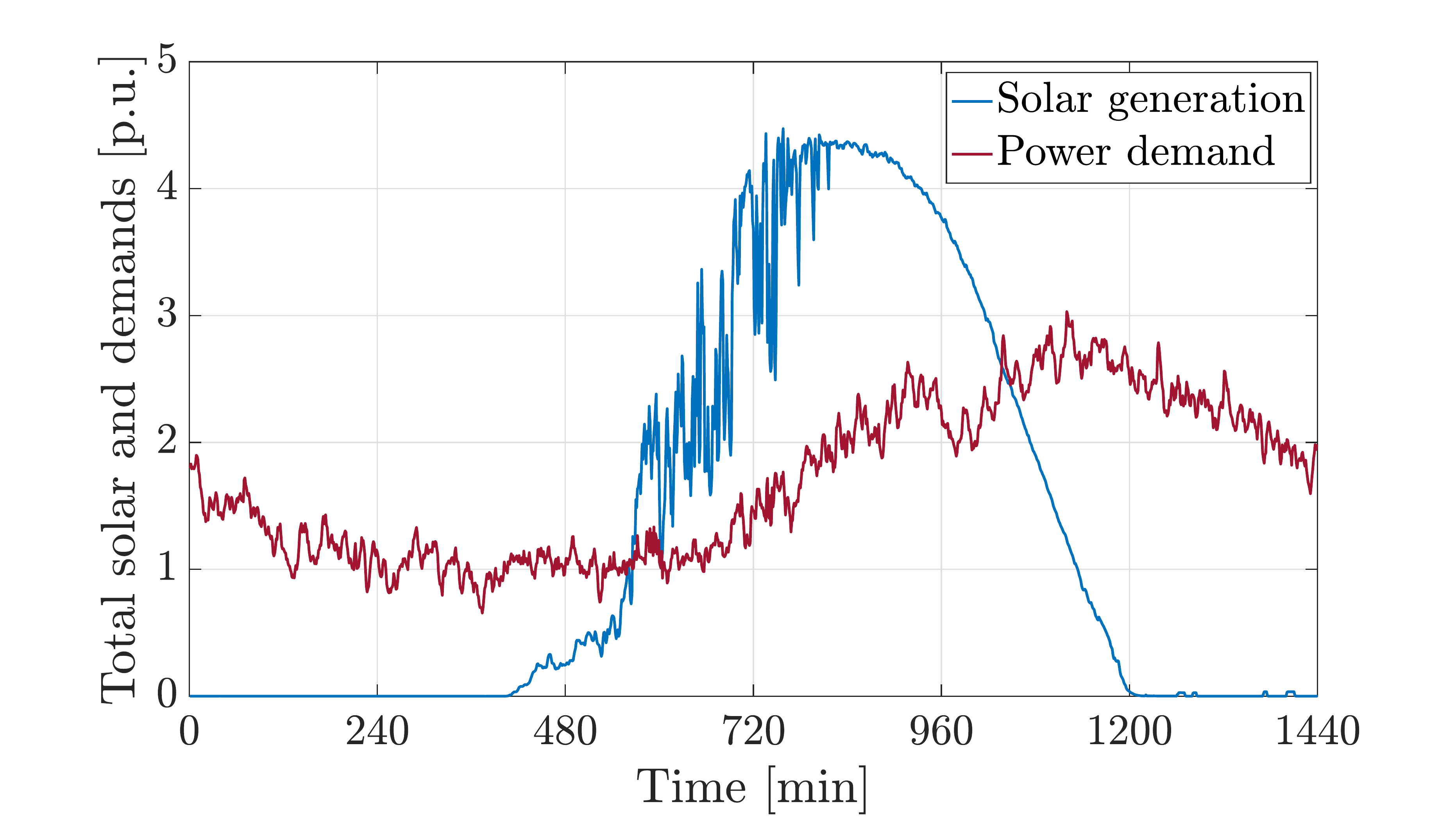}\\
    \caption{Minute-based data for the total (feeder-wise) solar power generation and active power demand.}
    \label{fig:solarload}
\end{figure}

\subsection{The Learning Stage} \label{subsec:simu-learning} 
We considered an optimization problem of the form \eqref{eq:ORPF} where the cost is set to
\begin{align*}
    f(\qbf_\Cc) &= \alpha \underbrace{\|\vbf(\qbf_\Cc) - \ones\|}_{\text{\ding{172}}} + (1-\alpha)\underbrace{(\qbf^\top \tilde \Rbf \qbf + \pbf^\top \tilde \Rbf \pbf)}_{\text{\ding{173}}},
\end{align*}
where \ding{172} and \ding{173} aims to minimize the voltage deviations and power losses~\cite{GC-RC:17}, respectively, and the parameter $\alpha$ trades-off those two objectives. 
We assume the 5 DERs have uniform generation capabilities, precisely, $\qbf_{\max} = 0.4 \times \ones$ MVAR and $\qbf_{\min} = -\qbf_{\max}$. The voltage limit vectors are set to $\vbf_{\max} = 1.05 \times \ones$ p.u. and $\vbf_{\min} = 0.95 \times \ones$ p.u.

The data set for the learning process is built using the aforesaid power demands and generations obtained from the Pecan Street data, which are intended as day-ahead forecasts.
We use the CVX toolbox~\cite{MG-SB:14-cvx} to solve the ORPF problem~\eqref{eq:ORPF} with linearized power flow~\eqref{eq:v=Rp+Xq}. Note however that, one can use any other power flow models to solve the ORPF problem.
We add the pseudo data points to the obtained data set as described in Remark~\ref{rmk:fic-points} with $\underline K = \overline K = 700$, which results in a total of $2840$ data points for each DER.
We implement the neural network approach according to Proposition~\ref{prop:universaL-approx} using TensorFlow 2.7.0 and conduct the training process in Google Colab with a single TPU with 32 GB memory. The number of episodes and the number of neurons $H$ are 2000 and 1000, respectively, and the neural networks are trained using the Adam optimizer~\cite{DPK-JB:15} with the learning rate initialized at 0.01 and decays every 500 steps with a base of 0.5. 

Fig.~\ref{fig:eq-functions-with-fic} plots the solutions to the ORPF problem~\eqref{eq:ORPF} for all data profiles, the equilibrium function $\phi_{32}$ learned with and without pseudo points, the standard droop function $\varrho_{32}$, and the optimized droop function $\rho_{32}$ for the DER at node 32 with $\alpha = \frac{1}{3}$.
In contrast to the case in which no pseudo points are added in the learning process, the learned equilibrium function with pseudo points reaches maximum reactive power compensation capability when voltage exceeds the limits. We further summarize in Table~\ref{tab:average_loss} the average loss for the whole training data set using the learned equilibrium functions and optimized/standard droop control functions, i.e.,
\begin{align}\label{eq:avg_loss}
    \frac{\sum^{K}_{k=1} \|\qbf_{\Cc,k}^\star - \square(\vbf_{\Cc,k}^\star) \|^2}{K \cdot C},
\end{align}
where $\square$ is $\phib$ for the data-based method, $\rhob$ for the optimized droop control, and $\varrhob$ for the standard droop control. The results illustrate the enhanced optimality of the learned equilibrium functions in approximating ORPF solutions compared to benchmarks.

\begin{table}[htb]
\centering
\caption{Average loss values for all data profiles}
\begin{tabular}{| c || c | c | c |}
\hline 
$\alpha$ & \textbf{Learned equil. func.} & Opt. droop func. & Std. droop func.\\
\hline
0 & $\mathbf{0.0090}$ & 0.0305 & 0.0438\\
\hline
1/3 & $\mathbf{0.0060}$ & 0.0171 & 0.0395\\
\hline
1/2 & $\mathbf{0.0071}$ & 0.0429 & 0.0662\\
\hline
2/3 & $\mathbf{0.0175}$ & 0.0724 & 0.0852\\
\hline
1 & $\mathbf{0.0377}$ & 0.0878 & 0.0965\\
\hline
\end{tabular}
\label{tab:average_loss}
\end{table}

Next, we illustrate the advantage of using the incremental algorithm. Recall that to guarantee the convergence of the non-incremental algorithm, i.e., $\epsilon = 1$, one needs to further enforce an additional slope constraint~\eqref{eq:static_global_conv} on the learned equilibrium functions, cf. Remark~\ref{rmk:slope-constraint}. Fig.~\ref{fig:loss} shows that this additional slope constraint leads to larger approximation errors of the learned equilibrium functions in fitting the data set (we do not consider the pseudo points during learning here for fairness), and thus degrades the optimality of system performance. 
\begin{figure}[ht]
    \centering
	\includegraphics[width=.52\textwidth]{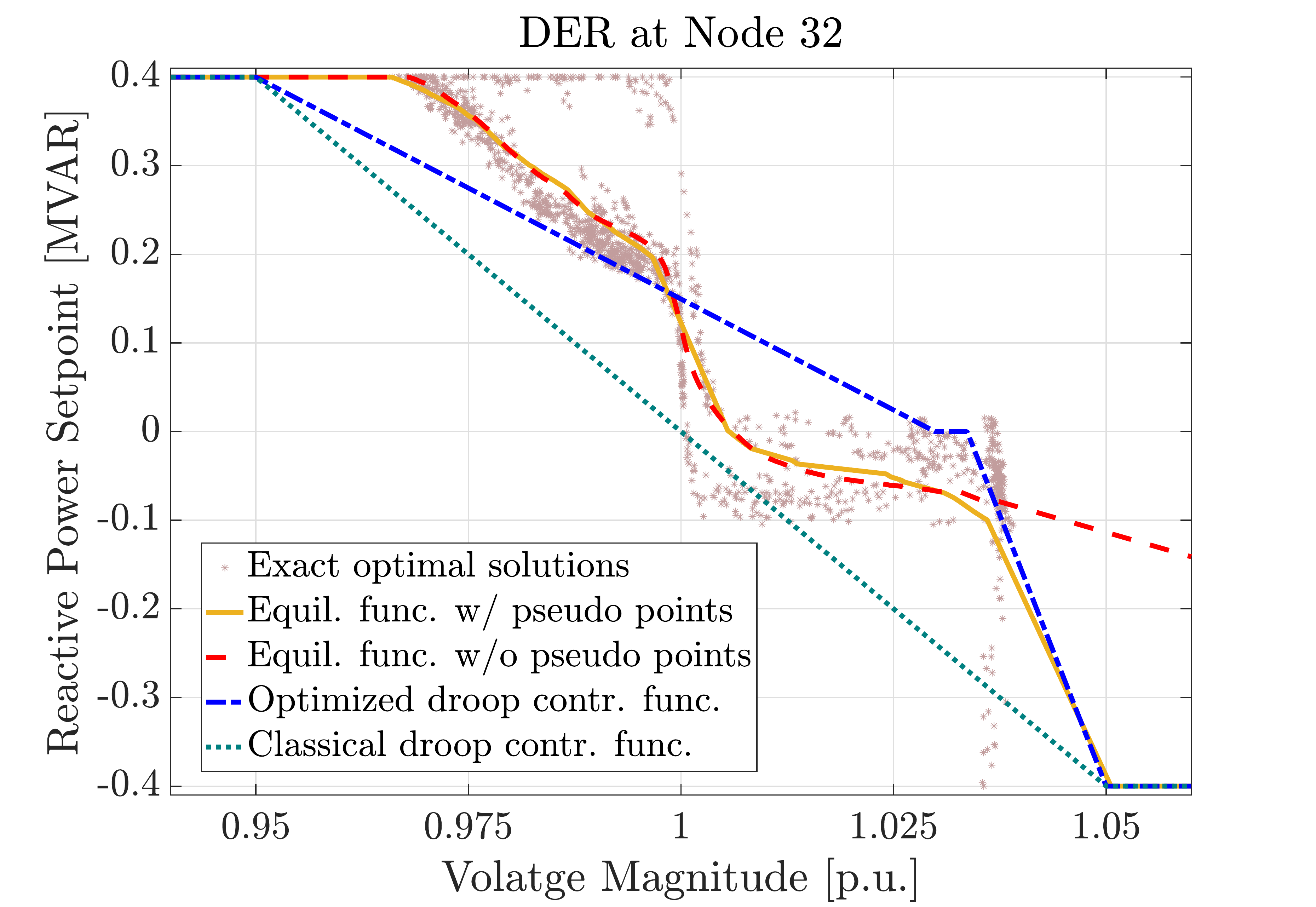}
  \caption{The orange solid and red dashed curves are, respectively, the learned equilibrium functions with and without considering pseudo data points for the DER at node 32. The blue dash-dotted and green dotted curves are, respectively, the optimized linear droop control function~\cite{HJ-CW-PL-JZ-GS-FD-JW:18} and standard linear droop function~\cite{GC-RC:17,KT-PS-SB-MC:11}.
  The comparison between the orange solid and red dashed curves illustrates the role of pseudo data points in learning equilibrium functions. The former reaches the maximum reactive power compensation capability when the voltage exceeds the limits, whereas the latter does not.}
  \label{fig:eq-functions-with-fic}
\end{figure}
\begin{figure}[t]
    \centering
    \includegraphics[width=.52\textwidth]{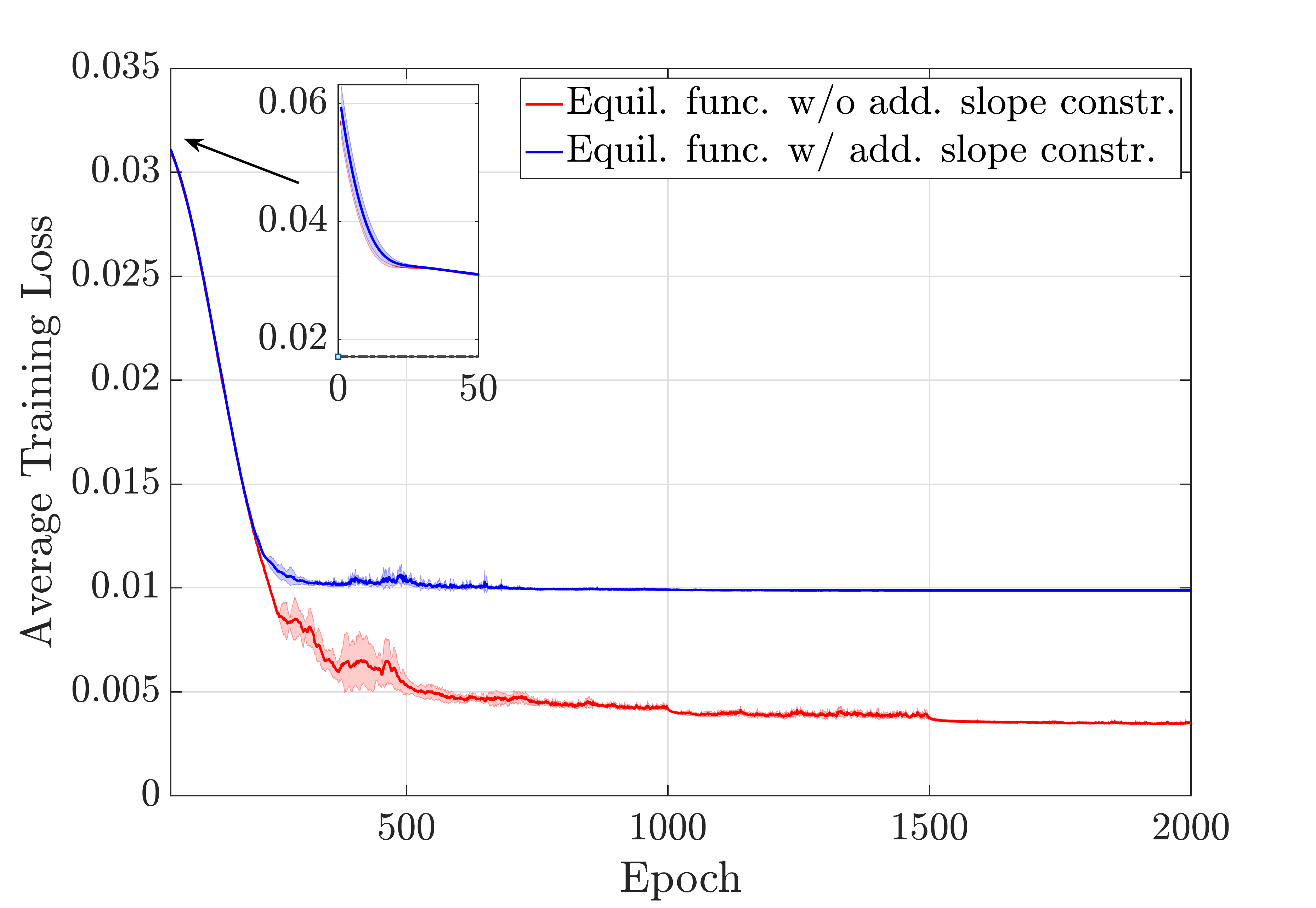}\\
    \caption{
    Comparison of average training losses (MSE) for all DERs with and and without the additional slope constraint~\eqref{eq:static_global_conv} on the equilibrium functions as described in Remark~\ref{rmk:slope-constraint}. The mean and standard deviations are evaluated based on five random seeds. This additional slope constraint leads to  a larger approximation error of the learned equilibrium functions in fitting the data points. 
    }
    \label{fig:loss}
\end{figure}

\subsection{The Control Stage}\label{subsec:simu-control} 
We run the following simulations using the learned equilibrium functions for the case $\alpha = \frac{1}{3}$ with the pseudo data points considered and assume that $\qbf_\Cc(0) = \zeros$. \textsc{Matpower}~\cite{RDQ-CEMS-RJT:11} was used to solve the power flow equation.
First, we verify the convergence properties of the proposed reactive power update rule~\eqref{eq:bus_react_upd} stated in Proposition~\ref{prop:convergence}. Consider the scenario where the load-generation profiles are fixed, Fig.~\ref{fig:stability} reports the evolution of the DERs' reactive power setpoints using load-generation profiles of the 695-th minute and considers 120 iterations of~\eqref{eq:bus_react_upd}. For $\epsilon = 0.369$, the reactive power setpoint trajectories converge to their final values, cf.~Fig.~\ref{fig:stability-0.369}, whereas the case $\epsilon = 1$ fails, cf.~Fig.~\ref{fig:stability-1}.
This is consistent with the sufficient condition 
$$0 < \epsilon < \min\Big\{1, \frac{2}{(\|\Xbf\|L + 1)^2}\Big\} = 0.3691$$ derived in Proposition~\ref{prop:convergence}.

Next, we test the proposed data-based control method in a scenario where the load-generation profiles are time varying. Specifically, we obtain load-generation profiles by randomly perturbing (5\%) the consumption data used to learn the equilibrium functions. This can be interpreted as having the data from the data set prescribing a day-ahead forecast, whereas their random perturbation acts as the true realization of the load-generation scenarios. These loads and generations are minute-based and we consider 120 iterations of~\eqref{eq:bus_react_upd} per minute with $\epsilon = 0.369$.
Fig.~\ref{fig:voltages} compares the evolution of the maximum/minimum voltages under the proposed data-based control method, the optimized droop control method, the ORPF solutions, and the case where no control action is taken.
One can observe that, in contrast to the proposed data-based method, the optimized droop control method induces instability issues within 12:00 and 16:00 causing voltages to oscillate.
We further compare in Fig.~\ref{fig:dist_power} the distances between the actual reactive power setpoints and ORPF solutions,
i.e., $\|\qbf_\Cc - \qbf_\Cc^\star\|$. Compared to the benchmarks, the reactive power setpoints during evolution under the proposed data-based method remains much closer to the ORPF solutions, thus leading to significantly improved optimality.
To further illustrate the effectiveness and advantages of the proposed data-based control method, Table~\ref{tab:dist_power} summarizes the comparison results of the proposed data-based control method against the standard and optimized linear droop control methods, as well as the case where no control action is taken for different values of $\alpha$. We quantify the performance by the average of the distances of the actual and optimal reactive power setpoints, i.e., average of $\|\qbf_\Cc(t) - \qbf_\Cc^\star(t)\|$ for the entire day.
It can be observed that the proposed data-based control method outperforms the benchmark methods in all cases. 

\begin{figure}[t]
    \centering
    \subfigure[$\epsilon = 0.369$]{
    \includegraphics[width=.52\textwidth]{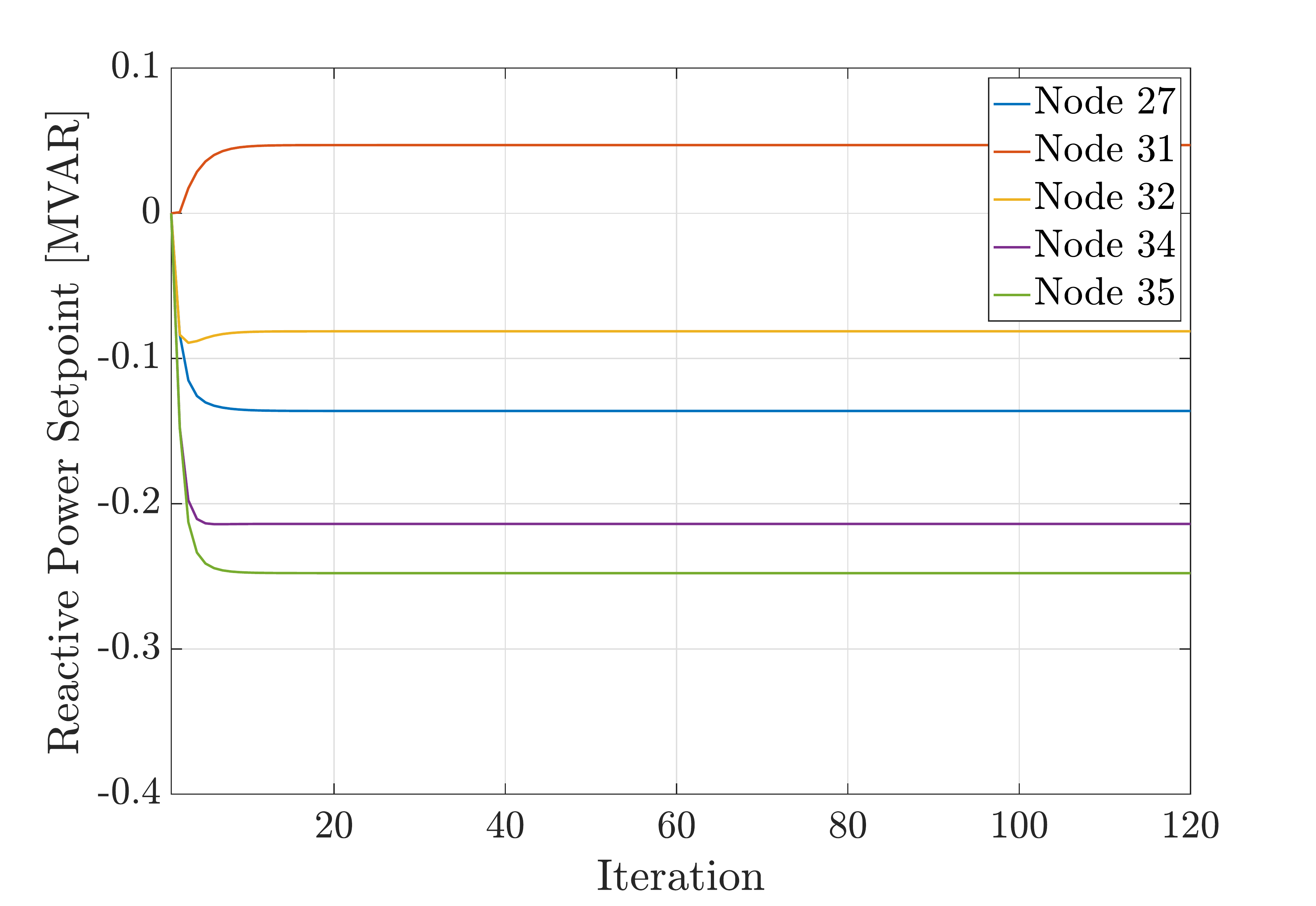}
    \label{fig:stability-0.369}
    }\\
    \vspace{-2ex}
    \subfigure[$\epsilon = 1$]{
    \includegraphics[width=.52\textwidth]{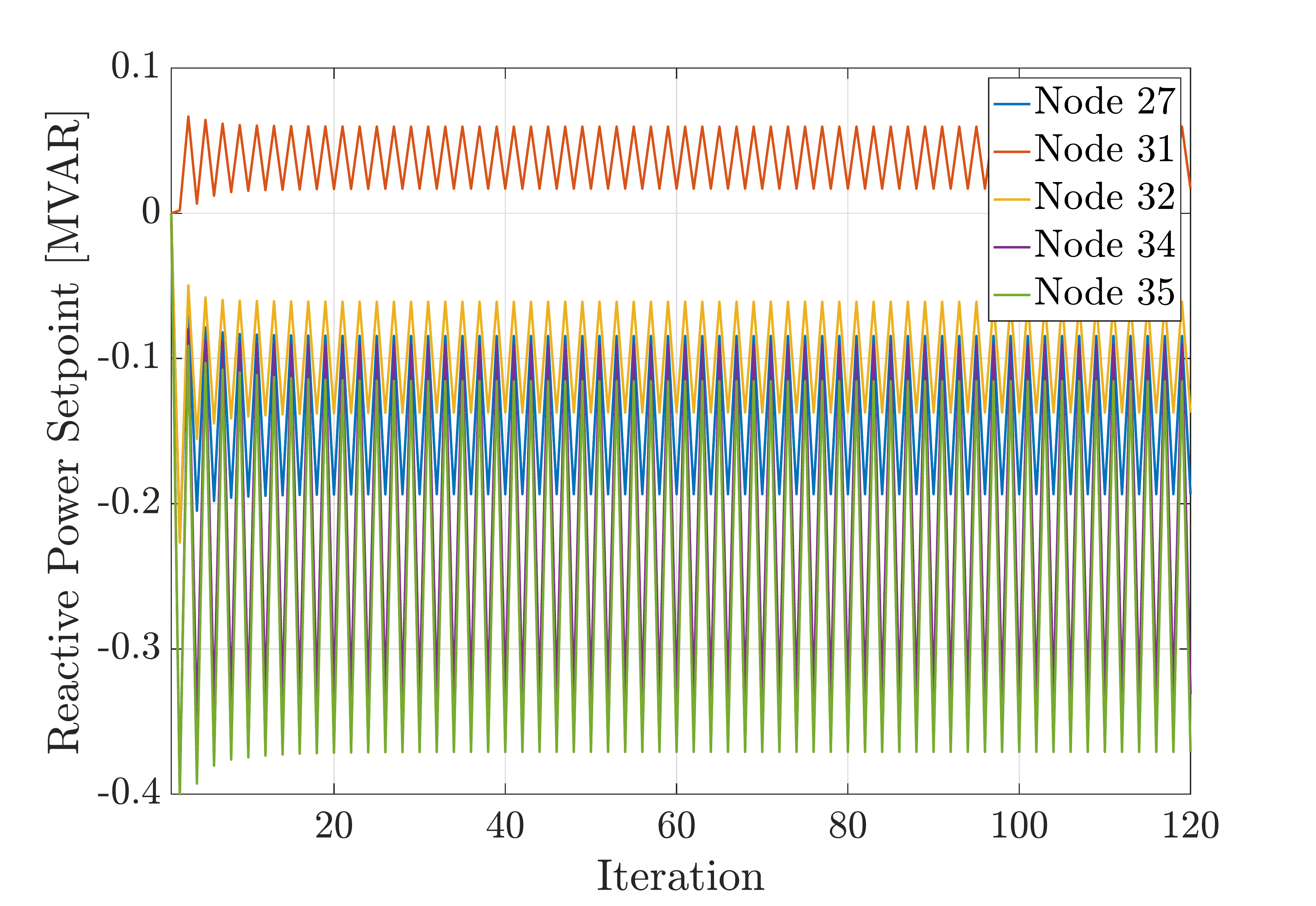}
    \label{fig:stability-1}
    }
  \caption{Evolution of reactive power setpoints under the proposed reactive power update rule~\eqref{eq:bus_react_upd} with (a) $\epsilon = 0.369$ and (b) $\epsilon = 1$, where we use the power data profiles of the 695-th minute and consider 120 iterations. This verifies the sufficient condition $0 < \epsilon < \min\{1,\frac{2}{(\|\Xbf\|L + 1)^2}\} = 0.3691$ in Proposition~\ref{prop:convergence} to ensure global asymptotic stability.
  }
  \label{fig:stability}
\end{figure}

\begin{figure}[htb]
    \centering
    \subfigure[Maximum Voltage]{
    \includegraphics[width=.52\textwidth]{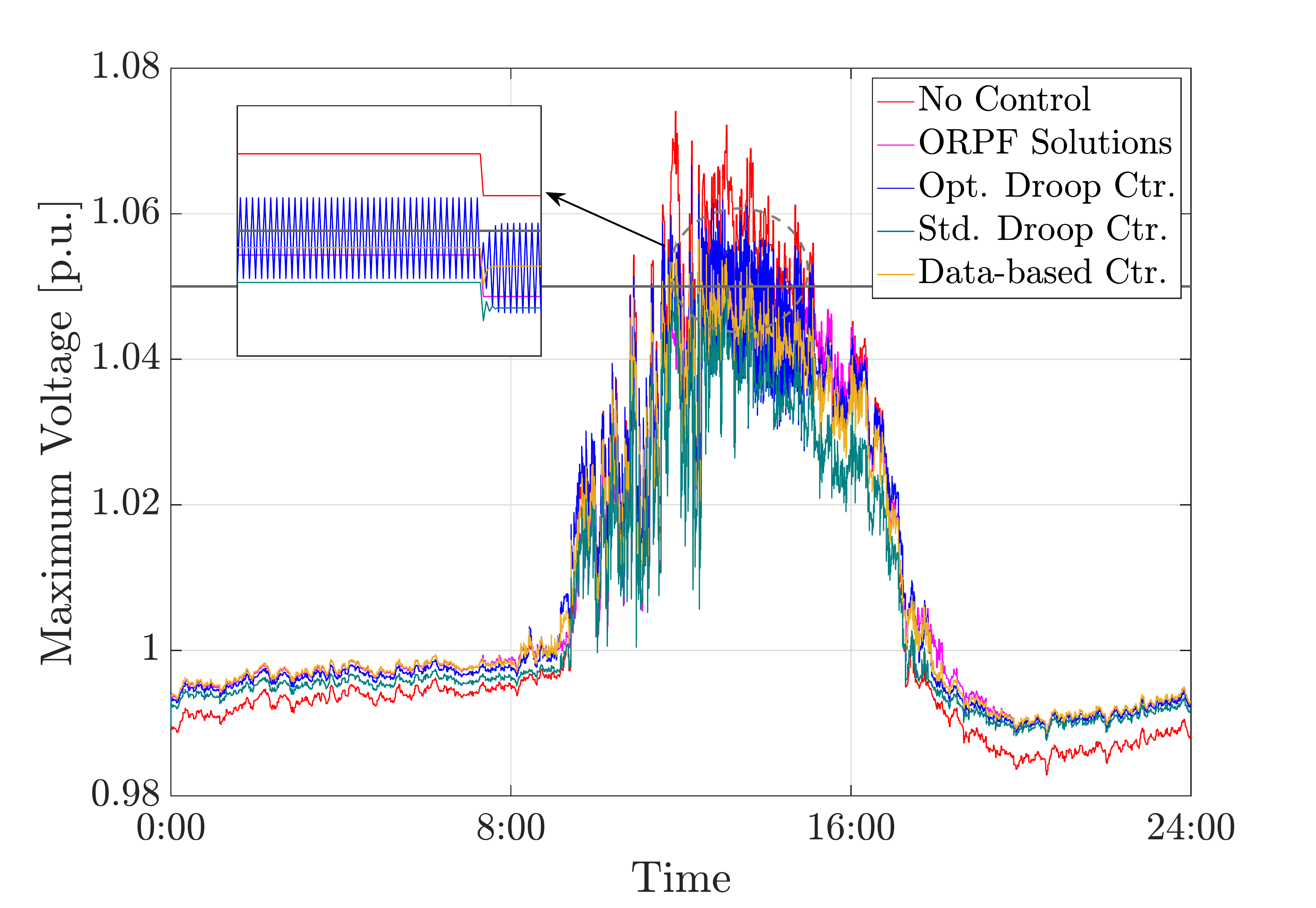}
    }\\
    \vspace{-2ex}
    \subfigure[Minimum Voltage]{
    \includegraphics[width=.52\textwidth]{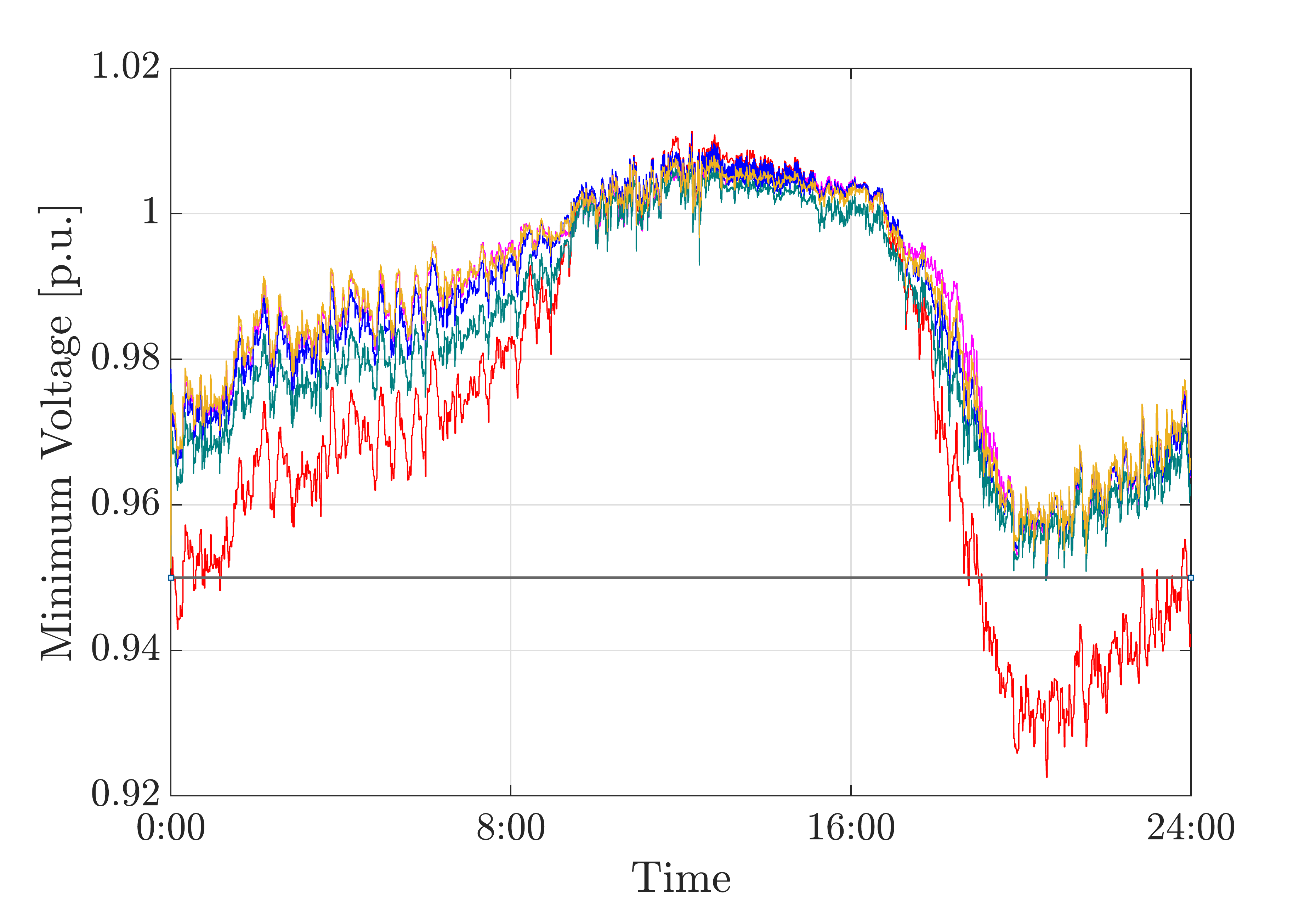}
    }
    \caption{Evolution of the maximum (top) and minimum (bottom) voltages of the IEEE 37-bus network under the proposed data-based, ORPF solutions, optimized linear droop, standard linear droop, and no control methods. For all minute-based data profiles, the ORPF problem is feasible and thus $\qbf_\Cc^\star$ always exists. The optimized droop control induces voltage instability issues, causing voltages oscillations during 12:00 and 16:00, while the proposed data-based method guarantees the convergence of voltages for every minute-based data profile.
    }
    \label{fig:voltages}
\end{figure}

\begin{figure}[htb]
    \centering
    \includegraphics[width=.52\textwidth]{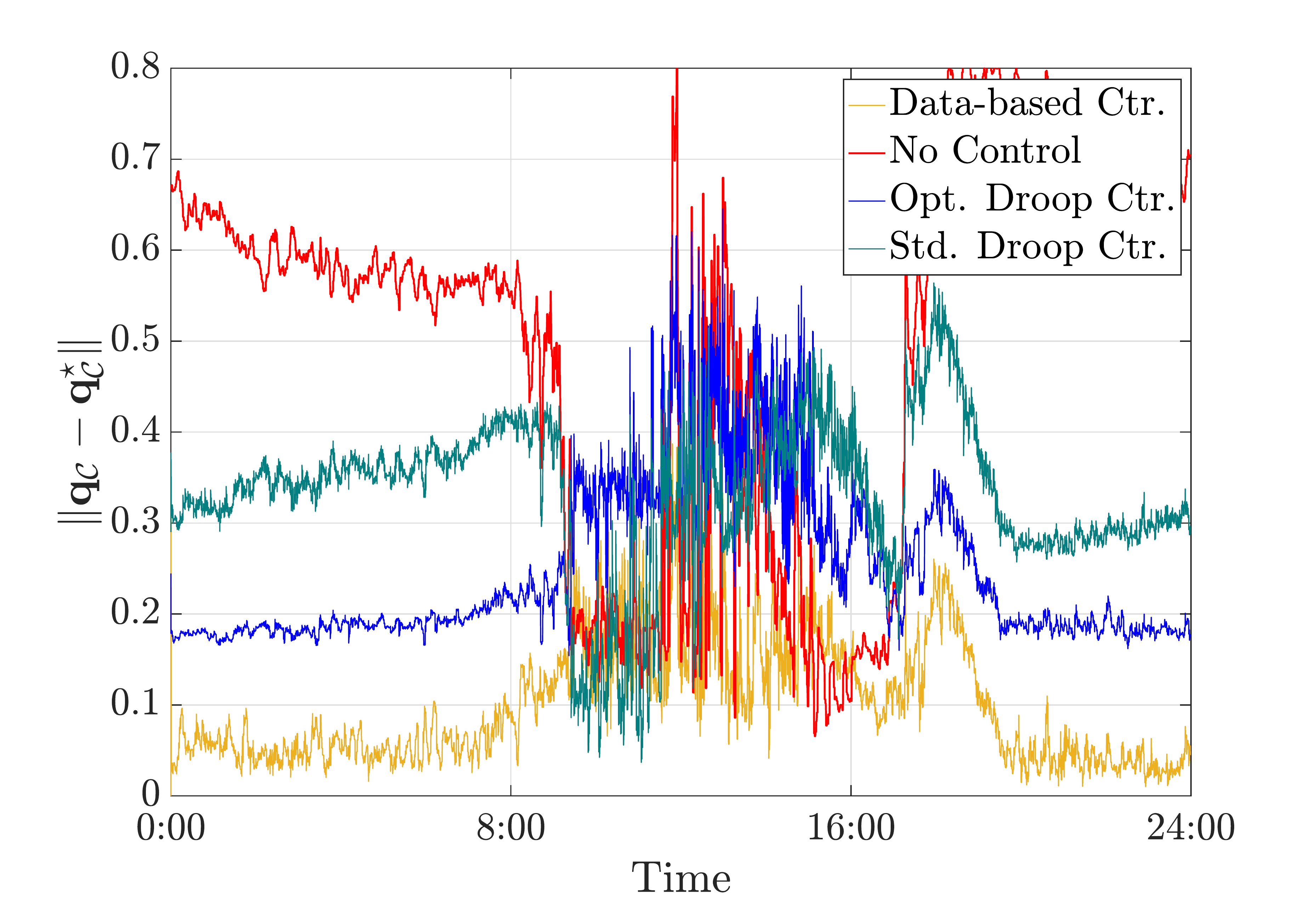}
    \caption{Evolution of the distance between actual reactive power setpoints and ORPF solutions of the IEEE 37-bus network under the proposed data-based, optimized linear droop, standard linear droop, and no control methods.
    }
    \label{fig:dist_power}
\end{figure}

\begin{table}[htb]
\centering
\caption{Average distances between actual reactive power setpoints and ORPF solutions for entire day}
\begin{tabular}{| c | c | c | c | c | c |}
\hline 
$\alpha$ & 0 & 1/3 & 1/2 & 2/3 & 1 \\
\hline
\hline
\textbf{Data-based} & $\mathbf{0.1185}$ & $\mathbf{0.0985}$ & $\mathbf{0.0784}$ & $\mathbf{0.1115}$ & $\mathbf{0.1645}$ \\
\hline
Opt. Droop & 0.2786 & 0.2474 & 0.3728 & 0.4410 & 0.4854 \\
\hline
Std. Droop & 0.2886 & 0.3311 & 0.4076 & 0.4699 & 0.5047 \\
\hline
No Control & 0.3160 & 0.5081 & 0.6842 & 0.8029 & 0.8358 \\
\hline
\end{tabular}
\label{tab:dist_power}
\end{table}

\subsection{Robustness to Voltage Measurement Noise} \label{subsec:robustness}
Here, we test the robustness of the proposed data-based method against the voltage measurement noise. Similar to Section~\ref{subsec:simu-control}, we use the learned equilibrium functions to run the simulations, but add random perturbations to the measurement of the local voltages. Table~\ref{tab:dist_noise} summarizes the distances between actual reactive power setpoints and the ORPF solutions under different voltage measurement noise levels. Specifically, 0.5\% perturbation to the voltage measurement corresponds to a common level of precision among smart meters in the United States~\cite{EEI-AEIC-UTC:11}, whereas the case of 1\% perturbation represents the biggest error allowed in power systems. Table~\ref{tab:dist_noise} indicates that, even in the latter case, the proposed data-based method still significantly outperforms the benchmark methods without measurement noise.

\begin{table}[t]
\centering
\caption{Average distances under measurement noise}
\begin{tabular}{|c | c | c | c | c | c |}
\hline
\diagbox{Noise}{$\alpha$} & 0 & 1/3 & 1/2 & 2/3 & 1 \\
\hline
0.2\%  & 0.1210 & 0.1017 & 0.0936 & 0.1291 & 0.1739 \\
\hline
0.5\% & 0.1317 & 0.1176 & 0.1378 & 0.1862 & 0.2237 \\
\hline
1.0\% & 0.1571 & 0.1553 & 0.2127 & 0.2762 & 0.3129 \\
\hline
\end{tabular}
\label{tab:dist_noise}
\end{table}

\subsection{Discussion}
Our simulation results above validate the improved performance of the proposed data-based method compared to the linear droop control method for different control goals. In fact, apart from considering the minimization of voltage deviations and power losses, our framework allows the users to consider any other type of cost functions, depending on specific control goals, to learn purely local controllers that steer system operating points to approximated ORPF solutions. 
However, as shown in Table~\ref{tab:average_loss}, different cost functions could result in different optimality gaps between the proposed data-based method and the ORPF approach. Since the data set we construct only maps the local voltage to the local optimal reactive power setpoint, it is possible that one fixed voltage corresponds to multiple optimal reactive power setpoints. On the other hand, it could also be that the optimal solution pairs are not as close to the non-increasing shape as we require the equilibrium functions to be.
We refer to these phenomena as \emph{data inconsistency}. We note that different selections of cost function significantly influence the data inconsistency, and thus leads to very different optimality gaps. For example, as Table~\ref{tab:average_loss} suggests, the data becomes significantly more inconsistent when the minimization of the voltage deviations takes a more important role in the cost function.
As part of our follow-up work, we plan to include other available local information to alleviate the data inconsistency challenge, e.g., prevailing (re)active power injections as additional inputs of the equilibrium function. Another important observation is that, although the ORPF approach strictly guarantees that the voltages are within limits, our approach does not. For instance, in Fig.~\ref{fig:voltages}, the voltage nadir during evolution under the proposed data-based method slightly violates the voltage limits. The reason is that when $\alpha$ is relatively small, many of the optimal solutions given by the ORPF problem lie on the boundary of the voltages limits. 
Since the local surrogates only provide approximations of the optimal solutions, the actual converged voltages can easily go out of limits in such situations.
On the other hand, as pointed out in~\cite{SB-RC-GC-SZ:19}, purely local control strategies generally have no guarantee on desired regulation, in the sense that the equilibrium $\qbf_\Cc^\sharp$ of~\eqref{eq:dyn_sys} could result in a $\vbf(\qbf_\Cc^\sharp) \notin [\vbf_{\min},\vbf_{\max}]$, even if there indeed exists $\qbf_\Cc$ such that $\vbf(\qbf_\Cc) \in [\vbf_{\min},\vbf_{\max}]$.

\section{Conclusions}\label{sec:conc}

We have presented a data-driven framework to design local Volt/Var controllers capable of steering a power distribution network towards efficient network configurations.
Building on the idea of learning local surrogates that map local voltages to reactive power setpoints that approximate the ORPF solution, we have proposed a local control update scheme and identified conditions on surrogates and control parameters so that the reactive power point converges in a global asymptotic sense.
By constructing a labeled data set of ORPF solutions with different load and generation profiles, we have trained neural networks whose resulting parameterized functions meet the conditions on surrogates by design to fit the data set. 
We have shown in AC power flow simulation tests that the proposed framework guarantees the voltage stability and significantly reduces the operation cost compared to prevalent local control approaches. Future research directions include considering the regulation of legacy devices and unbalanced DGs, enhancing data consistency by making use of other local information in building the data set, reducing the optimality gap during the learning process, and extending the proposed framework to a more general scenario where we take advantage of communication among neighboring agents. 

\appendices

\section{Technical Lemmas}

\begin{lemma}\longthmtitle{Bauer and Fike Theorem~\cite[Corollary 6.3.4]{RAH-CRJ:12}}\label{lem:bauer_and_fike}
Let $\Abf,\Ebf \in \real^{n \times n}$ with $\Abf$ normal. If $\hat{\lambda}$ is an eigenvalue
of $\Abf + \Ebf$, then there exists an eigenvalue $\lambda$ of $\Abf$ such that $|\hat{\lambda} - \lambda| \leq \|\Ebf\|$.
\end{lemma}

\begin{lemma}\longthmtitle{Positive semidefiniteness of $\Xbf \Mbf$ and upper bound of $\|\Xbf \Mbf\|$}\label{lem:eigvls}
The matrix $\Xbf \Mbf$ is positive semidefinite. Moreover, it holds
\begin{align*}
    \|\Xbf \Mbf\| \leq \|\Xbf\| L.
\end{align*}
\end{lemma}
\begin{proof}
Let $(\lambda_i, \xib_i)$ be a left eigenpair for $\Xbf\Mbf$. Then, $(\lambda_i, \xib_i\Xbf^{\frac 1 2})$ is a left eigenpair for the symmetric matrix $\Xbf^{\frac 1 2} \Mbf \Xbf^{\frac 1 2} \succeq 0$. Indeed,
\begin{align*}
   \xib_i \Xbf^{\frac 1 2} \Xbf^{\frac 1 2} \Mbf  \Xbf^{\frac 1 2} &= \xib_i \Xbf \Mbf \Xbf^{\frac 1 2}  
   = \lambda_i \xib_i  \Xbf^{\frac 1 2}.
\end{align*}
Therefore, $\Xbf\Mbf$ is positive semidefinite as well. 
Also,
\begin{align*}
    \|\Xbf \Mbf\| \leq \|\Xbf\| \|\Mbf\| \leq \|\Xbf\|L,
\end{align*}
which completes the proof.
\end{proof}

\section{Proofs of Propositions}

\begin{proof}[Proof of Proposition~\ref{prop:eq-point}]
First, we show by induction the feasibility of the reactive power update~\eqref{eq:bus_react_upd}. The initial power injection $q_n(0)$ belongs to $\Qc_n$ by hypothesis. Assume now that $q_n(t) \in \Qc_n$ and C3) holds. Then $q_n(t+1)$ is the convex combination of two elements of $\Qc_n$. 
Next, we show that~\eqref{eq:dyn_sys} has an unique equilibrium point. From~\eqref{eq:fixed-point}, the equilibrium exists if $\qbf_\Cc = \hbf(\qbf_\Cc)$ has a solution, where $\hbf : \Qc \rightarrow \Qc$ is a continuous vector function with $ \hbf(\qbf_\Cc) = \phib(\Xbf \qbf_\Cc + \hat \vbf_\Cc)$. Since $\Qc$ is convex and compact, according to Brouwer’s Fixed Point Theorem~\cite[Corollary 6.6]{KCB:85}, such a solution exists. Finally, to show uniqueness, we reason by contradiction. Assume both $(\qbf_\Cc^\sharp,\vbf_\Cc^\sharp)$ and $(\qbf_\Cc^\natural,\vbf_\Cc^\natural)$ are equilibrium points for~\eqref{eq:dyn_sys} with $\qbf_\Cc^\sharp \neq \qbf_\Cc^\natural$.
From~\eqref{eq:fixed-point-q},
\begin{align}
    \qbf_\Cc^\natural - \qbf_\Cc^\sharp &= \phib(\vbf_\Cc^\natural) - \phib(\vbf_\Cc^\sharp) = \Dbf(\vbf_\Cc^\natural - \vbf_\Cc^\sharp), \label{eq:proof_Prop}
\end{align}
where $\Dbf \in \real^{C \times C}$ is a diagonal matrix with
\begin{align*}
D_n = 
\begin{cases}
\frac{\phi_n(v_n^\natural) - \phi_n(v_n^\sharp)}{v_n^\natural - v_n^\sharp} & v_n^\natural \neq v_n^\sharp, \\
\hfil 0 & v_n^\natural = v_n^\sharp.
\end{cases}
\end{align*}
From C2), $\phi_n$ is nonincreasing in $v_n$ for all $n \in \Cc$. Hence, $D_n \leq 0, \forall n \in \Cc$, and $\Dbf \preceq 0$.
On the other hand,~\eqref{eq:fixed-point-v} yields 
$$\qbf_\Cc^\natural - \qbf_\Cc^\sharp = \Xbf^{-1} (\vbf_\Cc^\natural - \vbf_\Cc^\sharp).$$
Then, it follows that
$$(\Xbf^{-1} - \Dbf)(\vbf_\Cc^\natural - \vbf_\Cc^\sharp) = \zeros.$$
Since $\Xbf \succ 0$, it holds that $\Xbf^{-1} \succ 0$, $\Xbf^{-1} - \Dbf \succ 0$. As a consequence, $\vbf_\Cc^\natural - \vbf_\Cc^\sharp = \zeros$ and $\qbf_\Cc^\natural - \qbf_\Cc^\sharp = \zeros$, cf.~\eqref{eq:fixed-point-q}, which is a contradiction. This completes the proof.
\end{proof}

\begin{proof}[Proof of Proposition~\ref{prop:convergence}] 
Consider the voltage evolution under~\eqref{eq:dyn_sys},
\begin{align*}
&\vbf_\Cc(t+1) = \Xbf \qbf_\Cc(t+1) + \hat{\vbf}_\Cc 
\\
&= (1-\epsilon) \Xbf \qbf_\Cc(t) + \epsilon \Xbf \phib(\vbf_\Cc(t)) + (1-\epsilon)\hat{\vbf}_\Cc + \epsilon \hat{\vbf}_\Cc  
\\
&=(1-\epsilon) \vbf_\Cc(t) + \epsilon(\Xbf \phib( \vbf_\Cc(t) ) + \hat{\vbf}_\Cc) := \gbf(\vbf_\Cc(t)).
\end{align*}
We show that, for small enough values of $\epsilon$, the operator $\gbf: \real^C \to \real^C$ is a contraction,~i.e.,
\begin{align}\label{eq:contraction}
    \frac{\|\gbf(\vbf_\Cc) - \gbf(\vbf_\Cc^\prime)\|}{\|\vbf_\Cc - \vbf_\Cc^\prime\|} < 1 ,
\end{align}
for any $\vbf_\Cc,\vbf_\Cc^\prime \in \real^C$. Indeed, define a diagonal matrix $\Mbf \in \real^{C \times C}$ with the $n$-th diagonal entry being
\begin{align*}
    M_n = 
\begin{cases}    
\frac{ |\phi_n(v_n) - \phi(v_n^\prime)| }{ |v_n - v_n^\prime| } & v_n \neq v_n^\prime, \\
\hfil 0 & v_n = v_n^\prime.
\end{cases}
\end{align*}
Then, it follows that
\begin{align*}
    &\|\gbf(\vbf_\Cc) - \gbf(\vbf_\Cc^\prime)\| \\
    &= \big\|(1-\epsilon)(\vbf_\Cc-\vbf_\Cc^\prime) + \epsilon \Xbf \big(\phib(\vbf_\Cc) - \phib(\vbf_\Cc^\prime) \big)\big\| \notag \\
    &= \big\|(1-\epsilon) \sign(\vbf_\Cc-\vbf_\Cc^\prime) |\vbf_\Cc-\vbf_\Cc^\prime| \\
    & \qquad \qquad \qquad - \epsilon \Xbf \sign(\vbf_\Cc-\vbf_\Cc^\prime) |\phib(\vbf_\Cc) - \phib(\vbf_\Cc^\prime)| \big\| \notag \\
    &= \big\|(1-\epsilon) |\vbf_\Cc-\vbf_\Cc^\prime|  - \epsilon \Xbf |\phib(\vbf_\Cc) - \phib(\vbf_\Cc^\prime)| \big\| \notag \\
    &\leq \|(1-\epsilon) \Ibf - \epsilon \Xbf \Mbf \| \|\vbf_\Cc-\vbf_\Cc^\prime\|,
\end{align*}
where we have used in the second equality the fact that $\phi_n$ is nonincreasing in $v_n$ for each $n \in \Cc$, and thus $\sign(\phib(\vbf_\Cc) - \phib(\vbf_\Cc^\prime)) = - \sign(\vbf_\Cc-\vbf_\Cc^\prime)$.

To prove \eqref{eq:contraction}, it is sufficient to show that there always exists $\epsilon$ such that $\|(1-\epsilon) \Ibf - \epsilon \Xbf \Mbf \| < 1$, which is equivalent to proving that $\lambda_{\max}(\Gammab) < 1$, where $\Gammab \succeq 0$ is
\begin{align*}
  \Gammab \triangleq& \left[ (1-\epsilon) \Ibf - \epsilon \Xbf \Mbf \right]^\top \left[ (1-\epsilon) \Ibf - \epsilon \Xbf \Mbf \right] \\
=&  (1-\epsilon)^2 \Ibf - \epsilon(1-\epsilon) (\Mbf \Xbf^\top + \Xbf \Mbf) + \epsilon^2 \Mbf \Xbf^\top \Xbf \Mbf.
\end{align*}
Rewrite $\Gammab$ as
\begin{align*}
    \Gammab = &\underbrace{(1 - 2\epsilon) \Ibf - \epsilon(\Mbf \Xbf^\top + \Xbf \Mbf )}_{:=\Abf} \\
    &\qquad  + \epsilon^2 \underbrace{(\Ibf + \Xbf \Mbf + \Mbf \Xbf^\top + \Mbf \Xbf^\top \Xbf \Mbf)}_{:=\Ebf}.
\end{align*}
Note that $\Abf$ is symmetric. According to Lemma~\ref{lem:bauer_and_fike}, provided in the Appendix, it holds that $0 \leq \lambda_{\max}(\Gammab) \leq \lambda_{\max}(\Abf) + \epsilon^2\|\Ebf\|$. Now, it remains to show that there always exists $\epsilon$ such that
\begin{align*}
    \lambda_{\max}(\Abf) + \epsilon^2\|\Ebf\| < 1.
\end{align*}
Let $\underline{\lambda} = \lambda_{\min}(\Mbf \Xbf^\top + \Xbf \Mbf)$. According to Lemma~\ref{lem:eigvls}, $\Mbf\Xbf^\top, \Xbf\Mbf \succeq 0$, and therefore $\underline\lambda \geq 0$. This condition is then equivalent to
\begin{equation*}
 1 - \epsilon(2 + \underline{\lambda}) + \epsilon^2 \|\Ebf\| < 1,
\end{equation*}
which yields 
\begin{equation*}
0 < \epsilon < \frac{(2 + \underline\lambda)}{\|\Ebf\|}.
\label{eq:eps_first_bound}
\end{equation*}
Finally, according to Lemma~\ref{lem:eigvls}, it holds
\begin{align*}
    \|\Ebf\| &\leq 1 + 2 \|\Xbf\Mbf\| + \|\Xbf\Mbf\|^2 \\
    &\leq 1 + 2 \|\Xbf\|L + \|\Xbf\|^2L^2 = (\|\Xbf\|L + 1)^2,
\end{align*}
and hence
\begin{align*}
\frac{2 + \underline\lambda}{\|\Ebf\|} \geq \frac{2}{\|\Ebf\|} \geq \frac{2}{(\|\Xbf\|L + 1)^2}.
\end{align*}
Considering that we require $\epsilon \in [0,1]$,~\eqref{eq:conv_cond} follows, concluding  the proof.
\end{proof}
{\color{navy}
\noindent The proof of the Proposition~\ref{prop:convergence} here contains some imprecisions and we refer the readers to the Proposition 7.1 of the following book chapter: 
Z. Yuan, G. Cavraro, and J. Cort\'es, ``\href{https://www.sciencedirect.com/science/article/pii/B978044321524700013X}{Learning stable local Volt/Var controllers in distribution grids},'' \emph{Big Data Application in Power Systems (Second Edition)}, pp. 135-159, 2024, where we provide a corrected version that leads to the same conclusions but poses less restrictive condition on the stepsize $\epsilon$. 
}

\begin{proof}[Proof of Proposition~\ref{prop:universaL-approx}]
To show the sufficiency of~\eqref{eq:enc_c2}, notice that, for any $x \in \real$, if $x < b_1$, then $\mathsf{N}(x) = \beta$; and for $x \geq b_1$, $\mathsf{N}(x)$ is divided into $H$ segments with the slope of the $J$-th segment being $\sum_{j=1}^{J} w_j$ for $J \in \{1,2,\dots,H\}$. Given~\eqref{eq:enc_c2}, it follows that $\mathsf{N}$ is nonincreasing. To show the necessity of~\eqref{eq:enc_c2}, suppose $\mathsf{N}$ is nonincreasing. Then, there exists a segment of the equilibrium function, e.g., the $J$-th, $J \in \{1,2,\dots,H\}$ which is increasing. Hence, $\sum_{j=1}^{J} w_j > 0$ which is a  contradiction. 

We next show the universal approximation property.
Given a compact domain $\Xc$ of the form $\underline x \leq x \leq \overline x$, consider an equispaced partition of $\Xc$ into $H$ intervals, with the length of each interval being $s = \frac{\overline x - \underline x}{H}$. 
Consider the function
\begin{align*}
\mathsf{N}(x) = \sum_{h=1}^{H} w_h {\rm ReLU}(x - b_h) + \beta,
\end{align*}
whose parameters are defined as: $\beta = g(\underline x)$, $b_h = \underline x + (h-1)s$, $w_1 = \frac{g(\underline x + s)-g(\underline x)}{s}$, $w_2 = \frac{g(\underline x+2s)-g(\underline x+s)}{s} - w_1$,..., $w_h = \frac{g(\underline x+hs)-g(\underline x+(h-1)s)}{s} - \sum_{k=1}^{h-1} w_{k}$. Note that it holds that 
$$\sum_{k=1}^h w_k = \frac{g(\underline x+hs)-g(\underline x+(h-1)s)}{s} \leq 0$$ 
because $g(x)$ is nonincreasing, for $h \in \{1,2,\dots,H\}$. Hence, inequality~\eqref{eq:enc_c2} is satisfied. It can be observed that
$$ |\mathsf{N}(x) - g(x)| \leq L_g s $$
where $L_g$ is the Lipschitz constant of $g$. Hence, by setting $H \geq \frac{L_g(\overline x - \underline x)}{\eta} = O(1/\eta)$, it holds that $|\mathsf{N}(x) - g(x)| \leq L_g s \leq \eta$ for all $x \in \Xc$. This completes the proof.
\end{proof}

\bibliographystyle{ieeetr}

\begin{IEEEbiography}[{\includegraphics[width=1in,height=1.25in,clip,keepaspectratio]{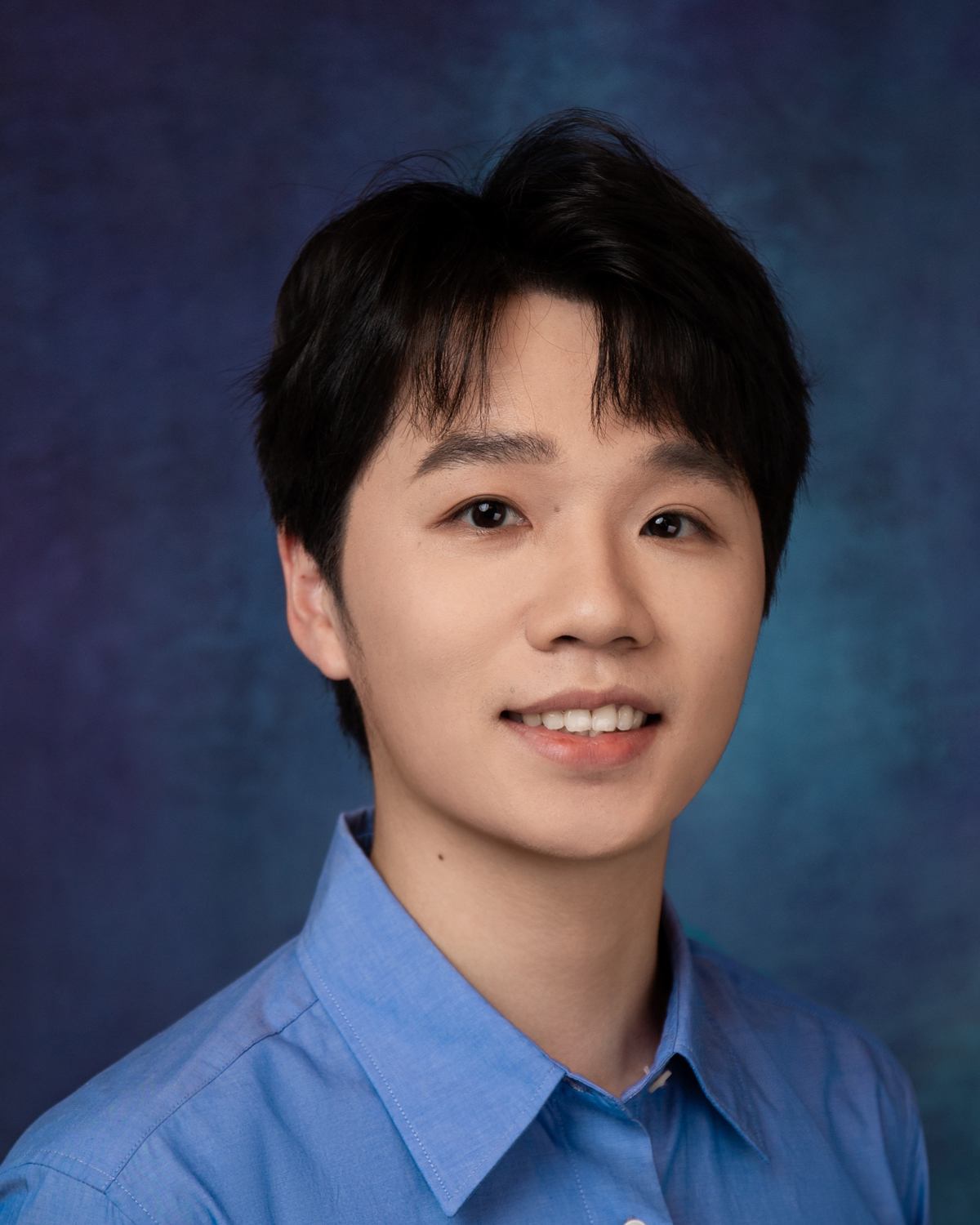}}]{Zhenyi Yuan} (Student Member, IEEE) was born in Jiangxi, China, in 1998. He is currently working toward the Ph.D. degree in mechanical engineering at the University of California, San Diego, CA, USA. Prior to that, he was enrolled in successive undergraduate and postgraduate program at the Honors School of Harbin Institute of Technology, Harbin, China, where he received the B.S. and M.S. degrees in control science and engineering in 2018 and 2020, respectively. He was also a visiting research assistant with the Department of Information Engineering of the Chinese University of Hong Kong, Hong Kong, in 2021. His research interests lie at the intersection of control, optimization and learning, with applications to smart grids and robotic systems.
\end{IEEEbiography}

\begin{IEEEbiography}%
[{\includegraphics[width=1in,height=1.25in,clip,keepaspectratio]{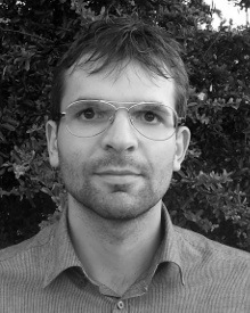}}]
{Guido Cavraro} (Member, IEEE) received the Ph.D. degree in Information Engineering from the University of Padova, Italy, in 2015. He was a visiting scholar at the California Institute for Energy and Environment (CIEE) at U.C. Berkeley in 2014. In 2015 and 2016, he was a postdoctoral associate at the Department of Information Engineering of the University of Padova. From 2016 to 2018, he was a postdoctoral associate with the Bradley Department of Electrical and Computer Engineering of Virginia Tech, USA. Currently, he is a Senior Researcher with the Power Systems Engineering Center at National Renewable Energy Laboratory, USA. His research interests include control, optimization, and estimation with applications to power systems.
\end{IEEEbiography}

\begin{IEEEbiography}[{\includegraphics[width=1in,height=1.25in,clip,keepaspectratio]{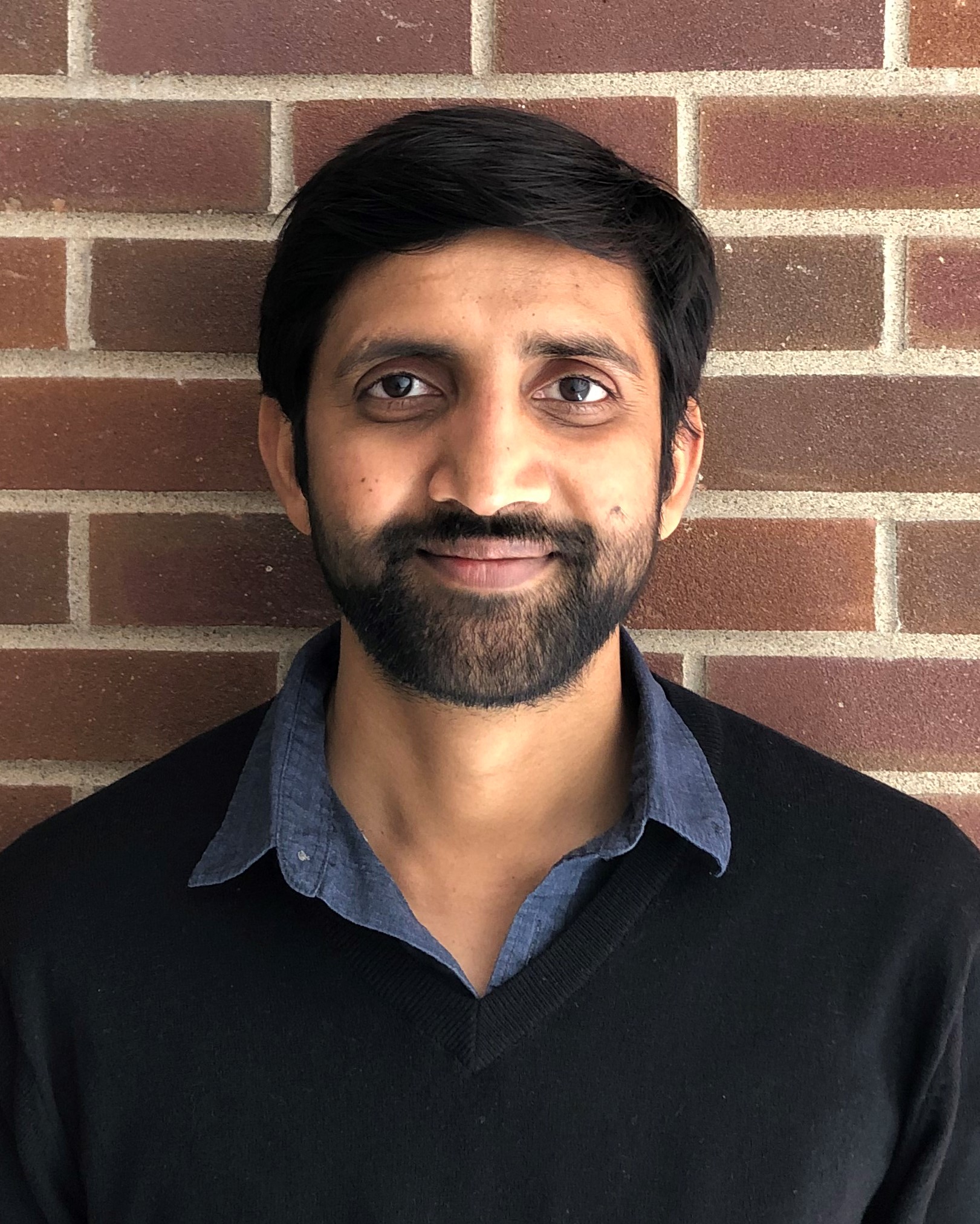}}] {Manish K. Singh} (Member, IEEE) received the B.Tech. degree from the Indian Institute of Technology (BHU), Varanasi, India, in 2013; and the M.S. and Ph.D. degrees in electrical engineering from Virginia Tech, Blacksburg, VA, USA, in 2018 and 2021, respectively. He is currently a postdoctoral researcher with the University of Minnesota, Minneapolis, MN. During 2013-2016, he worked as an Engineer in the Smart Grid Dept. of POWERGRID, the central transmission utility of India. His research interests are focused on optimization, control, and learning techniques to develop algorithmic solutions for the operation and analysis of electric power systems as well as water and natural gas networks.
\end{IEEEbiography}

\begin{IEEEbiography}[{\includegraphics[width=1in,height=1.25in,clip,keepaspectratio]{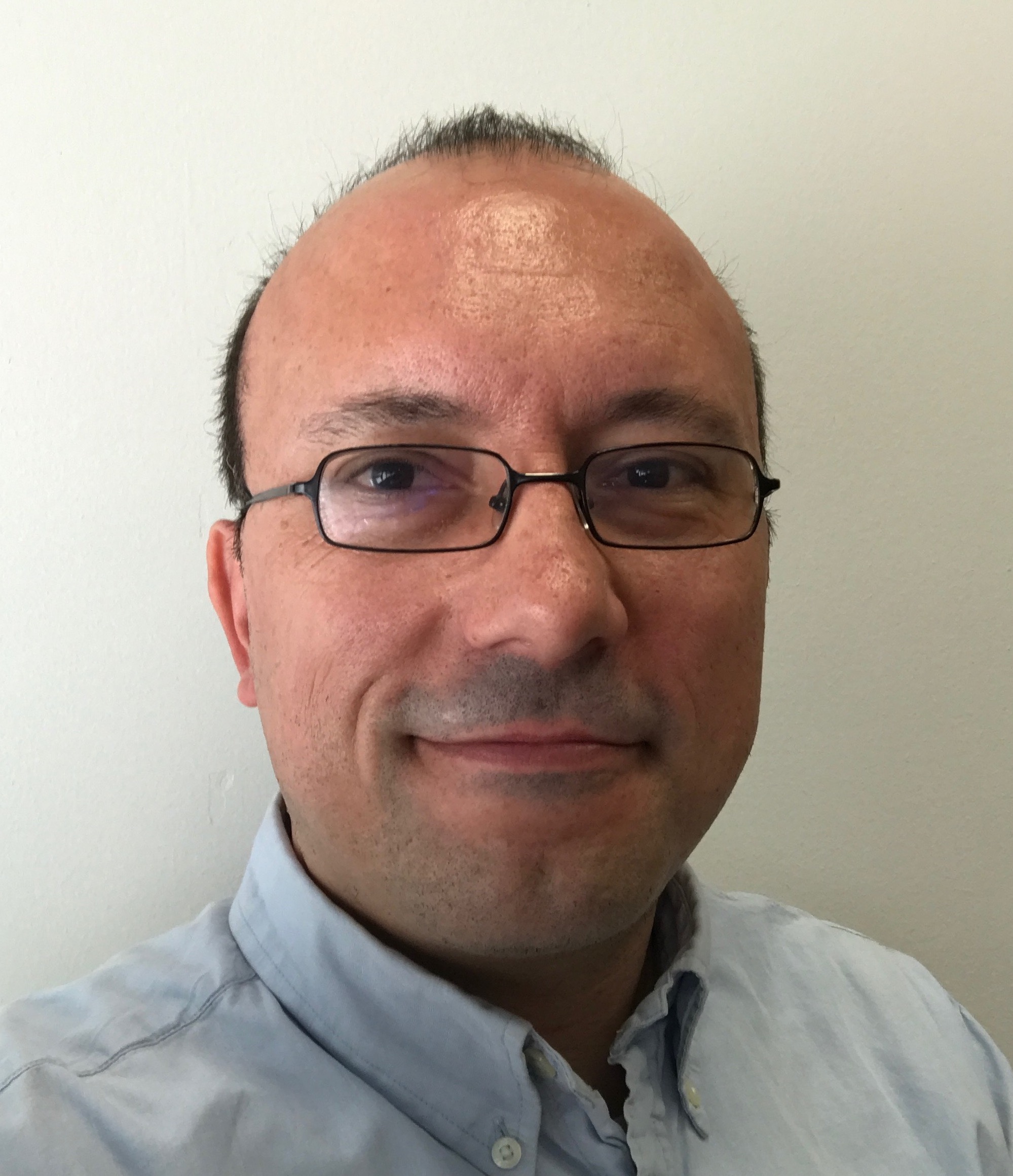}}]{Jorge Cort\'es} (Fellow, IEEE) received the Licenciatura degree in mathematics from
Universidad de Zaragoza, Zaragoza, Spain, in 1997, and the
Ph.D. degree in engineering mathematics from Universidad Carlos III de
Madrid, Madrid, Spain, in 2001. He held postdoctoral positions with
the University of Twente, Twente, The Netherlands, and the University
of Illinois at Urbana-Champaign, Urbana, IL, USA. He was an Assistant
Professor with the Department of Applied Mathematics and Statistics,
University of California, Santa Cruz, CA, USA, from 2004 to 2007. He
is a Professor in the Department of Mechanical and Aerospace
Engineering, University of California, San Diego, CA, USA.  He is a
Fellow of IEEE, SIAM, and IFAC.  His research interests include
distributed control and optimization, network science, nonsmooth
analysis, reasoning and decision making under uncertainty, network
neuroscience, and multi-agent coordination in robotic, power, and
transportation networks.
\end{IEEEbiography}

\end{document}